\documentclass{elsarticle}
\usepackage{amsmath,amssymb,graphicx,enumerate,bbm}
\usepackage{url,wasysym}
\newcommand{\assign}{:=}
\newcommand{\emdash}{---}
\newcommand{\mathd}{\mathrm{d}}
\newcommand{\op}[1]{#1}
\newcommand{\tmdfn}[1]{\textbf{#1}}
\newcommand{\tmem}[1]{{\em #1\/}}
\newcommand{\tmmathbf}[1]{\ensuremath{\boldsymbol{#1}}}
\newcommand{\tmop}[1]{\ensuremath{\operatorname{#1}}}

\newcommand{\tmstrong}[1]{\textbf{#1}}
\newcommand{\tmtextit}[1]{{\itshape{#1}}}
\newcommand{\tmtextsf}[1]{{\sffamily{#1}}}
\newcommand{\tmtexttt}[1]{{\ttfamily{#1}}}
\newcommand{\tmtextup}[1]{{\upshape{#1}}}
\newcommand{\um}{-}
\newcommand{\upl}{+}
\newtheorem{definition}{Definition}
\newtheorem{theorem}{Theorem}
\newenvironment{proof}{\vspace{1ex}\noindent{\bf Proof}\hspace{0.5em}}
	{\hfill\qed\vspace{1ex}}

\newcommand{\defEq}{\ensuremath{\assign}}
\newcommand{\arr}{\ensuremath{\rightarrow}}
\newcommand{\C}{\ensuremath{\mathfrak{C}}}
\newcommand{\uc}{uniformly continuous}

\newcommand{\lift}[1]{\ensuremath{\fastMap  #1}}
\newcommand{\ballsym}{\ensuremath{\tmmathbf{B}}}
\newcommand{\ball}[4]{\ensuremath{\ballsym^{#1}_{#2}  #3  #4}}
\newcommand{\Prop}{\ensuremath{\op{\star}}}
\newcommand{\pure}[1]{\ensuremath{\widehat{#1}}}
\newcommand{\app}{@}
\newcommand{\Map}{\ensuremath{\mars}}

\newcommand{\fastMap}{\ensuremath{\tmop{\mathsf{map}}}}
\newcommand{\fastBind}{\ensuremath{\tmop{\mathsf{bind}}}}
\newcommand{\maptwo}[1]{\ensuremath{\left\langle #1 \right\rangle}}
\newcommand{\foldmaptwo}[1]{\ensuremath{\left\{ #1 \right\}}}
\newcommand{\Q}{\ensuremath{\mathbbm{Q}}}
\newcommand{\R}{\ensuremath{\mathbbm{R}}}

\newcommand{\complete}{\ensuremath{\mathfrak{C}}}
\newcommand{\SF}{\ensuremath{\mathfrak{S}}}
\newcommand{\BF}{\ensuremath{\mathfrak{B}}}
\newcommand{\IF}{\ensuremath{\mathfrak{I}}}
\newcommand{\comp}{\text{{\tmstrong{\tmtextup{B}}}}}
\newcommand{\flip}{\text{{\tmstrong{\tmtextup{C}}}}}
\newcommand{\id}{\text{{\tmstrong{\tmtextup{I}}}}}
\newcommand{\const}{\text{{\tmstrong{\tmtextup{K}}}}}
\newcommand{\diag}{\text{{\tmstrong{\tmtextup{W}}}}}
\newcommand{\ess}{\text{{\tmstrong{\tmtextup{S}}}}}
\newcommand{\idzeroone}{\ensuremath{\id_{[0, 1]}}}
\newcommand{\glue}[3]{\ensuremath{#1 \vartriangleright #2 \vartriangleleft #3}}
\newcommand{\splitl}[2]{\ensuremath{#1 \blacktriangleright #2}}
\newcommand{\splitr}[2]{\ensuremath{#2 \blacktriangleleft #1}}
\newcommand{\ou}{\ensuremath{(0,1)_{\mathbbm{Q}}}}
\newcommand{\coq}{\tmtexttt{Coq}}
\begin{document}
\title{A computer verified, monadic, functional implementation of the integral.}

\author[ru]{Russell O'Connor}
\author[ru]{Bas Spitters}
\address[ru]{Radboud University Nijmegen}
\begin{abstract}
  We provide a computer verified exact monadic functional implementation of
  the Riemann integral in type theory. Together with previous work by
  O'Connor, this may be seen as the beginning of the realization of Bishop's
  vision to use constructive mathematics as a programming language for exact
  analysis.
\end{abstract}
\maketitle
\section{Introduction}
Integration is one of the fundamental techniques in numerical computation.
However, its implementation using floating point numbers requires continuous
effort on the part of the user in order to ensure that the results are
correct. This burden can be shifted away from the end-user by providing a
library of {\tmem{exact}} analysis in which the computer handles the error
estimates. For high assurance we use computer verified proofs that the
implementation is actually correct; see~{\cite{typesreal-article}} for an
overview. It has long been suggested that by using {\tmem{constructive
mathematics}} exact analysis and provable correctness can be
unified~{\cite{Bishop67,Bishop:num}}. Constructive mathematics provides a high
level framework for specifying computations
(Section~\ref{ss:constructiveMath}). However, Bishop~{\cite{Bishop67}} p.357
writes:

\begin{quotation}
  {\tmem{As written, this book is person-oriented rather than
  computer-oriented. It would be of great interest to have a computer-oriented
  version. Without such a version, it is hard to predict with any confidence
  what form computer-oriented abstract analysis will eventually assume. A
  thoughtful computer-oriented presentation should uncover many interesting
  phenomena.}}
\end{quotation}

Our aim is to provide such a presentation for Riemann integration. In fact, we
provide much more. We provide an implementation in dependent type theory
(Section~\ref{ss:typeTheory}). Type theory is a formal framework for
constructive mathematics~{\cite{ITT,CMCP,NPS}}. It supports the development of
formal proofs, while, at the same time, being an efficient functional
programming language with a dependent type system. We use the
Coq~{\cite{Coq,BC04}} proof assistant, which is an implementations of the
Calculus of Inductive Constructions (CIC)~{\cite{CoquandHuet,CoquandPaulin}}.
\ However, we believe that the ideas presented in this paper are general
enough to easily be developed in other implementations of type theory, such as
Martin-L\"of type theory{\footnote{In particular we do not believe that we
make any essential use of impredicativity of propositions in
Coq.}}{\cite{ITT,CMCP,NPS}}, so our presentation is mostly done in a
type-theoretic agnostic way.

Coq includes a compiler~{\cite{Compiler}} based on OCaml's virtual machine to
allow efficient evaluation{\footnote{We copy the conclusions from the
benchmarks carried out in~{\cite{Compiler}}:`...our reducer runs about as fast
as OCaml's bytecode interpreter; the speed ratio varies between 1.4 and 0.95.
Compiling the extracted Caml code with the OCaml native-code compiler results
in speed ratios between 3.5 and 5.6, which is typical of the speed-ups
obtained by going from bytecode interpretation to native-code generation.'}}.
As a feasibility study, we have implemented Riemann integration. Our
implementation is functional and structured in a monadic way. This structure
greatly simplifies the integrated development of the program together with its
correctness proof.

In constructive analysis one approximates real numbers by rational, or dyadic
numbers. Rational numbers, as opposed to the real numbers, can be represented
exactly in a computer. The real numbers are the completion of the rationals.
The completion construction can be organized in a monad, a familiar construct
from functional programming (Section~\ref{ss:completion-monad}). This
completion monad provides an efficient combination of proving and
computing~{\cite{OConnor:mscs}}. In this paper, we use a similar technique:
the integrable functions are in the completion of rational step functions
(Section~\ref{ss:Step}), and the same monadic implementation is reused.

Our contributions include:
\begin{itemize}
  \item We show that the step functions form a monad itself
  (Section~\ref{StepF-mon}) that distributes over the completion monad
  (Section~\ref{ss:dist-monad}).
  
  \item Using the applicative functor interface of the step function monad we
  lift functions and relations to step functions
  (Section~\ref{ss:applicative}).
  
  \item Using combinators we also lift theorems to reason about these
  functions and relations on step functions
  (Section~\ref{ss:liftingTheorems}).
  
  \item We define both $L^1$ and $L^{\infty}$ metrics on step functions
  (Section~\ref{metric-step}) and define integration on the completion of the
  $L^1$ space (Section~\ref{ss:IFBF}).
  
  \item We show how to embed uniformly continuous functions into this space in
  order to integrate them (Section~\ref{ss:riemann}).
  
  \item We extend our definition of Riemann integral to a Stieltjes integral
  (Section~\ref{ss:stieltjes}).
\end{itemize}
\subsection{Notation}We will use traditional notation from functional
programming for this paper. Thus $f x$ will represent function application. We
will typically use curried functions, so $f x y$ will represent $(f x) y$, and
$f$ will have type $X \Rightarrow Y \Rightarrow Z$ (meaning $X \Rightarrow (Y
\Rightarrow Z)$).

We will mostly gloss over details about equivalence relations for types. We
will use $\asymp$ to represent the equivalence relation to be used with the
types in question. We will use $\assign$for defining functions and constants.

We denote the type of the closed unit interval as $[0, 1]$, and $] 0, 1 [$
will be the type of the open interval. We denote the the open interval
restricted to the rational numbers by $] 0, 1 [_{\mathbbm{Q}}$.

\section{Background}

\subsection{\label{ss:constructiveMath}Constructive mathematics and type
theory}We wish to use constructive reasoning because constructive proofs have
a computational interpretation. For example, a constructive proof of $\varphi
\vee \psi$ tells which of the two disjuncts hold. A proof of $\exists n :
\mathbbm{N}. P n$ gives an explicit value for $n$ that makes $P n$ hold. Most
importantly, we have a functional interpretation of $\Rightarrow$ and
$\forall$. A proof of $\forall n : \mathbbm{N}. \exists m : \mathbbm{N}. R n
m$ is interpreted as a function with an argument $n$ that returns an $m$
paired with a proof of $R n m$. A proof of $\neg \varphi$, which is equal to
$\varphi \Rightarrow \bot$ by definition, is a function taking an arbitrary
proof of $\varphi$ to a proof of $\bot$ (false){\emdash}which means there
should not be any proofs of $\varphi$.

The connectives in constructive logic come equipped with their constructive
rules of inference (given by natural deduction){\cite{Sorensen}}. Excluded
middle ($\varphi \vee \neg \varphi$) cannot be deduced in general, and proof
by contradiction, $\neg \neg \varphi \Rightarrow \varphi$, is also not
provable in general.

\subsection{\label{ss:typeTheory}Dependently typed functional programming}The
functional interpretation of constructive deductions is given by the
Curry-Howard isomorphism~{\cite{Sorensen}}. This isomorphism associates
formulas with dependent types, and proofs of formulas with functional programs
of the associated dependent types. For example, the identity function $\lambda
x : A. x$ of type $A \Rightarrow A$ represents a proof of the tautology $A
\Rightarrow A$. Table~\ref{C-H} lists the association between logical
connectives and type constructors.

\begin{table}[h]\begin{center}
  \begin{tabular}{|l|l|}
    \hline
    {\tmstrong{Logical Connective}} & {\tmstrong{Type Constructor}}\\
    \hline
    implication: $\Rightarrow$ & function type: $\Rightarrow$\\
    \hline
    conjunction: $\wedge$ & product type: $\times$\\
    \hline
    disjunction: $\vee$ & disjoint union type: $+$\\
    \hline
    true: $\top$ & unit type: $()$\\
    \hline
    false: $\bot$ & void type: $\emptyset$\\
    \hline
    for all: $\forall x. P x$ & dependent function type: $\Pi x. P x$\\
    \hline
    exists: $\exists x. P x$ & dependent pair type: $\Sigma x. P x$\\
    \hline
  \end{tabular}\end{center}
  \caption{\label{C-H}The association between formulas and types given by the
  Curry-Howard isomorphism.}
\end{table}

In dependent type theory, functions from values to types are allowed. Using
types parametrized by values, one can create dependent pair types, $\Sigma x :
A. P x$, and dependent function types, $\Pi x : A. P x$. A dependent pair
consists of a value $x$ of type $A$ and a value of type $P x$. The type of the
second value depends on the first value, $x$. A dependent function is a
function from the type $A$ to the type $P x$. The type of the result depends
on the value of the input.

The association between logical connectives and types can be carried over to
constructive mathematics. We associate mathematical structures, such as the
natural numbers, with inductive types in functional programming languages. We
associate atomic formulas with functions returning types. For example, we can
define equality on the natural numbers, $x =_{\mathbbm{N}} y$, as a recursive
function:
\begin{eqnarray*}
  0 =_{\mathbbm{N}} 0 & \assign & \top\\
  S x =_{\mathbbm{N}} 0 & \assign & \bot\\
  0 =_{\mathbbm{N}} S y & \assign & \bot\\
  S x =_{\mathbbm{N}} S y & \assign & x =_{\mathbbm{N}} y
\end{eqnarray*}
One catch is that general recursion is not allowed when creating functions.
The problem is that general recursion allows one to create a fixed-point
operator, $\tmop{\mathsf{fix}} : (\varphi \Rightarrow \varphi) \Rightarrow
\varphi$, that corresponds to a proof of a logical inconsistency. To prevent
this, we allow only well-founded recursion over an argument with an inductive
type. Because well-founded recursion ensures that functions always terminate,
the language is not Turing complete. However, one can still express
fast-growing functions, such as the Ackermann function, without difficulty by
using higher-order functions~{\cite{Thompson:1991}}.

Because proofs and programs are written in the same language, we can freely
mix the two. For example, in previous work~{\cite{OConnor:mscs}}, the real
numbers are presented by the type
\begin{equation}
  \exists f : \mathbbm{Q}^{\upl} \Rightarrow \mathbbm{Q}. \forall
  \varepsilon_1 \varepsilon_2 . |f \varepsilon_1 - f \varepsilon_2 | \leq
  \varepsilon_1 + \varepsilon_2 . \label{R}
\end{equation}
A value of this type is a pair of a function $f : \mathbbm{Q}^{\upl}
\Rightarrow \mathbbm{Q}$ and a proof of $\forall \varepsilon_1 \varepsilon_2 .
|f \varepsilon_1 - f \varepsilon_2 | \leq \varepsilon_1 + \varepsilon_2$. The
idea is that a real number is represented by a function $f$ that maps any
requested precision $\varepsilon : \mathbbm{Q}^+$ to a rational approximation
of the real number. Not every function of type $\mathbbm{Q}^{\upl} \Rightarrow
\mathbbm{Q}$ represents a real number. Only those functions that have coherent
approximations should be allowed. The proof object paired with $f$ witnesses
the fact that $f$ has coherent approximations. This is one example of how
mixing functions and formulas allows one to create precise data-types.

\subsection{Extensional Equality}In this paper, we will use the
equality sign ($\op{=}$) for {\tmdfn{extensional equality}}.
Two functions $f, g$ of the same type are considered extensionally equal when,
for any input given to both functions, the outputs of the functions are
extensionally equal:
\[ f = g \defEq \forall a. f (a) = g (a) . \]
Two values of an inductive type are extensionally equal when their
constructors are the same and all parameters are extensionally equal.

Extensional equality is the finest equality we will need. However, {\coq} uses
a finer equality called intensional equality for its fundamental equality.

Another sort of equality that we will frequently use is setoid equality (see
Section~\ref{setoid}), which is generally coarser than extensional equality.

\subsection{Setoids Instead of Quotients}\label{setoid}A quotient type is a
type modulo a given equivalence relation on that type. For instance, the type
$\mathbbm{Q}$ is often considered as a quotient of the type $\mathbbm{Z}
\times \mathbbm{N}^+$. Coq does not have quotient types. One reason for this
is that it would destroy the decidability of type checking. One instead passes
around the equivalence relation in question. To do this, one often uses a data
structure called a setoid, or a Bishop set~{\cite{Bishop67,Hofmann,Capretta}}.
A setoid $(A, {\nobreak} \op{\asymp_A})$ is a type paired with an equivalence
relation on that type. Functions between setoids that preserve their
equivalence relations are called {\tmdfn{respectful}}. Proving that a function
is respectful consists of the same work in traditional mathematics needed to
prove that a function over quotients is well-defined. Respectful functions are
also called {\tmdfn{morphisms}}.

\subsubsection{Rewrite Automation}\label{RewriteDatabase}Coq supports reasoning about setoids
through its tactics \tmtexttt{setoid\_rewrite} and
\tmtexttt{setoid\_replace}~{\cite{Coen:2004}}. These tactics will
automatically create the deductions for substitution of setoid equivalent
terms into respectful functions and relations. This support makes reasoning
about setoid equivalence almost as easy as reasoning about equality in Coq.

Furthermore, Coq has the ability to define a database of rewrite lemmas. These
lemmas have terms of the form $a \asymp_A b$ for their conclusions. When they
are added to the database the user indicates which way substitution should be
performed (the same lemma can be added to different databases with different
directions). The user can then use the database as a rewrite system to process
a hypothesis or goal. The \tmtexttt{autorewrite {\tmem{<database>}}} tactic
will repeatedly try to use the lemmas in the named database to rewrite the
goal. Well crafted rewrite databases can be used to quickly transform or
simplify expressions.

\subsection{Metric spaces}\label{ss:MS}Traditionally, a metric space is
defined as a set $X$ with a metric function $d : X \times X \Rightarrow
\mathbbm{R}^{0 \upl}$ satisfying certain axioms. The usual constructive
formulation requires $d$ be a computable function. In previous
work~{\cite{OConnor:mscs}}, it was useful to take a more relaxed definition
for a metric space that does not require the metric be a function. A similar
construction can be found in the work by Richman~{\cite{Richman:2008}}.
Instead, the metric is represented via a (respectful) ball relation $\ballsym
: \mathbbm{Q}^+ \Rightarrow X \Rightarrow X \Rightarrow \star$, where $\star$
is the type of propositions, satisfying five axioms:
\begin{enumerate}
  \item $\forall x \varepsilon . \ball{}{\varepsilon}{x}{x}$
  
  \item $\forall x y \varepsilon . \ball{}{\varepsilon}{x}{y} \Rightarrow
  \ball{}{\varepsilon}{y}{x}$
  
  \item $\forall x y z \varepsilon_1 \varepsilon_2 .
  \ball{}{\varepsilon_1}{x}{y} \Rightarrow \ball{}{\varepsilon_2}{y}{z}
  \Rightarrow \ball{}{\varepsilon_1 + \varepsilon_2}{x}{z}$
  
  \item $\forall x y \varepsilon . (\forall \delta . \varepsilon < \delta
  \Rightarrow \ball{}{\delta}{x}{y}) \Rightarrow \ball{}{\varepsilon}{x}{y}$
  
  \item $\forall x y. (\forall \varepsilon . \ball{}{\varepsilon}{x}{y})
  \Rightarrow x \asymp y$
\end{enumerate}
The ball relation {\ball{}{{\varepsilon}}{x}{y}} expresses that the points $x$
and $y$ are within $\varepsilon$ of each other. We call this a ball
relationship because the partially applied relation $\ballsym^X_{\varepsilon}
x : X \Rightarrow \star$ is a predicate that represents the closed ball of
radius $\varepsilon$ around the point $x$.

For example, $\mathbbm{Q}$ can be equipped with the usual metric by defining
the ball relation as
\[ \ball{\mathbbm{Q}}{\varepsilon}{x}{y} \assign |x - y| \leq \varepsilon . \]
This definition satisfies all the required axioms.

\subsection{Uniform continuity}We are interested in the category of metric
spaces with uniformly continuous functions between them. A function $f : X
\Rightarrow Y$ between two metric spaces is {\tmdfn{uniformly continuous with
modulus}} $\mu_f : \mathbbm{Q}^+ \Rightarrow \mathbbm{Q}^+$ if
\[ \forall x_1 x_2 \varepsilon . \ball{X}{\mu_f \varepsilon}{x_1}{x_2}
   \Rightarrow \ball{Y}{\varepsilon}{(f x_1)}{(f x_2)} \text{.} \]

A function is {\tmdfn{uniformly continuous}} if it is uniformly continuous
with some modulus. We use the notation $X \rightarrow Y$ with a single bar
arrow to denote the type of uniformly continuous functions from $X$ to $Y$.
This record type consists of three parts, a function $f$ of type $X
\Rightarrow Y$, a modulus of continuity, and a proof that $f$ is uniformly
continuous with the given modulus. We will leave the projection to the
function type implicit and allow us to write $f x$ when $f : X \rightarrow Y$
and $x : X$. Our definition of uniform continuity implies that the function is
respectful.

\subsection{Monads}Moggi~{\cite{moggi:1989}} recognized that many non-standard
forms of computation may be modeled by monads{\footnote{In category theory one
would speak about the Kleisli category of a (strong) monad.}}.
Wadler~{\cite{Wadler92a}} popularized their use in functional programming.
Monads are now an established tool to structure computation with side-effects.
For instance, programs with input $X$ and output $Y$ which have access to a
mutable state $S$ can be modeled as functions of type $X \times S \Rightarrow
Y \times S$, or equivalently $X \Rightarrow (Y \times S)^S$. The type
constructor $\mathfrak{M}Y \assign (Y \times S)^S$ is an example of a monad.
Similarly, partial functions may be modeled by maps $X \Rightarrow Y_{\bot}$,
where $Y_{\bot} \assign Y + ()$ is a monad. The reader monad, $\mathfrak{M}Y
\assign Y^E$, for passing an environment implicitly will play an important
role in this paper.

The formal definition of a (strong) monad is a triple $(\mathfrak{M},
\tmop{\mathsf{return}}, \tmop{\mathsf{bind}})$ consisting of a type
constructor $\mathfrak{M}$ and two functions:
\begin{eqnarray}
  \tmop{\mathsf{return}} & : & X \Rightarrow \mathfrak{M}X \nonumber\\
  \tmop{\mathsf{bind}} & : & (X \Rightarrow \mathfrak{M}Y) \Rightarrow \mathfrak{M}X
  \Rightarrow \mathfrak{M}Y \nonumber
\end{eqnarray}
We will denote $( \tmop{\mathsf{return}} x)$ as $\hat{x}$, and $(
\tmop{\mathsf{bind}} f)$ as $\check{f}$. These two operations must satisfy the
following laws:
\begin{eqnarray}
  \tmop{\mathsf{bind}}\ \tmop{\mathsf{return}}\ a & \asymp & a \nonumber\\
  \check{f}  \hat{a}  & \asymp & f a \nonumber\\
  \check{f} ( \check{g} a) & \asymp & \tmop{\mathsf{bind}} ( \check{f} \circ
  g) a \nonumber
\end{eqnarray}
Alternatively, we can define a (strong) monad using three functions:
\begin{eqnarray}
  \tmop{\mathsf{return}} & : & X \Rightarrow \mathfrak{M}X \nonumber\\
  \tmop{\mathsf{map}} & : & (X \Rightarrow Y) \Rightarrow (\mathfrak{M}X
  \Rightarrow \mathfrak{M}Y) \nonumber\\
  \tmop{\mathsf{join}} & : & \mathfrak{M}(\mathfrak{M}X) \Rightarrow
  \mathfrak{M}X \nonumber
\end{eqnarray}
satisfying certain laws. These can be obtained from the previous presentation
of a monad by defining
\begin{eqnarray}
  \tmop{\mathsf{map}} f m & \assign & \tmop{\mathsf{bind}} (
  \tmop{\mathsf{return}} \circ f) m \nonumber\\
  \tmop{\mathsf{join}} m & \assign & \check{\id} m. \nonumber
\end{eqnarray}
where {\id} is the identity function. Conversely, given the $(
\tmop{\mathsf{return}}, \tmop{\mathsf{map}}, \tmop{\mathsf{join}})$
presentation we define
\begin{eqnarray}
  \tmop{\mathsf{bind}} f & \assign & \tmop{\mathsf{join}} \circ (
  \tmop{\mathsf{map}} f) . \nonumber
\end{eqnarray}
\subsection{Completion monad\label{ss:completion-monad}}The first monad that
we will meet in this paper is O'Connor's completion monad
$\complete$~{\cite{OConnor:mscs}}. Given a metric space $X$, the completion of
$X$ is defined by
\[ \complete X \assign \exists f : \mathbbm{Q}^{\upl} \Rightarrow X. \forall
   \varepsilon_1 \varepsilon_2 . \ball{X}{\varepsilon_1 + \varepsilon_2}{(f
   \varepsilon_1)}{(f \varepsilon_2)} . \]
The real numbers defined as the completion, $\mathbbm{R} \assign \complete
\mathbbm{Q}$, is exactly the type given in equation~\ref{R}.

The function $\tmop{\mathsf{return}} : X \rightarrow \complete X$ is the
embedding of a metric space in its completion. The function
$\tmop{\mathsf{join}} : \complete ( \complete X) \rightarrow \complete X$ is
half of this isomorphism between $\complete ( \complete X)$ and $\complete X$
(with return being the other half). Finally, a uniformly continuous function
$f : X \rightarrow Y$ can be lifted to operate on complete metric spaces,
$\tmop{\mathsf{map}} f : \complete X \rightarrow \complete Y$. Uniformly
continuity is essential in this definition of $\tmop{\mathsf{map}}$. This
means that {\complete} is a monad on the category of metric spaces with
uniformly continuous functions. One advantage of this approach is that it
helps us to work with simple representations. To specify a function from
$\mathbbm{R} \rightarrow \mathbbm{R}$, one can simply define a uniformly
continuous function $f : \mathbbm{Q} \rightarrow \mathbbm{R}$, and then
$\check{f} : \mathbbm{R} \rightarrow \mathbbm{R}$ is the required function.
Hence, the completion monad allows us to do in a structured way what was
already folklore in constructive mathematics: to work with simple, often
decidable, approximations to continuous objects; see
e.g.~{\cite{Schwichtenberg}}.

\section{Informal Presentation of Riemann Integration}

In this section, we present our work in informal constructive mathematics.
Everything presented here has been formalized in Coq, except where otherwise
noted.

We will implement Riemann integration as follows:
\begin{enumerate}
  \item Define step functions;
  
  \item Introduce applicative functors and show that step functions form an
  applicative functor;
  
  \item Show that the step functions form a metric space under both the $L^1$
  and $L^{\infty}$ norms;
  
  \item Define integrable functions as the completion of the step functions
  under the $L^1$ norm;
  
  \item Define integration first on step functions and lift it to operate on
  integrable functions;
  
  \item Define an injection from the uniformly continuous functions to the
  integrable functions in order to integrate them.
\end{enumerate}
At the end, we will see that it is natural to generalize our Riemann integral
to a Stieltjes integral.

\subsection{Step functions}\label{ss:Step}\label{ss:InductiveDef}Our first
goal will be to define (formal) step functions and some important
operations on them. For any type $X$, we first define the
inductive data type of (rational) step functions from the unit interval to
$X$, denoted by $\SF X$. A step function is either a constant function,
$\tmop{\mathsf{const}} x$, for some $x : X$, or two
step functions, $f : \SF X$ and $g : \SF X$ glued at a point in $o$,
$\tmop{\mathsf{glue}} o f g$, where $o$ must be a
rational number strictly between 0 and 1. We will sometimes write $( \tmop{\mathsf{const}} x)$ as
$\hat{x}$, and $(\tmop{\mathsf{glue}} o f g)$ as $\glue{f}{o}{g}$.

\begin{definition}
  The rules for constructing the inductive data type {\SF}:
          \[\frac{x : X}{\tmop{\mathsf{const}} x : \SF (X)}\qquad
  \frac{o : \ou \quad f : \SF (X)\quad g :\SF (X)}{\glue{f}{o}{g:{\SF}(X)}}\]
\end{definition}

The elements of this inductive type are intended to be interpreted as step
functions on $[0, 1]$. The interpretation of $\hat{x}$ is the constant
function on $[0, 1]$ returning $x$. The interpretation of {\glue{f}{o}{g}} is
$f$ squeezed into the interval $[0, o]$ and $g$ squeezed into the interval
$[o, 1]$. In this sense $f$ and $g$ are ``glued'' together.

\begin{figure}[h]\begin{center}
  \resizebox{3in}{!}{\includegraphics{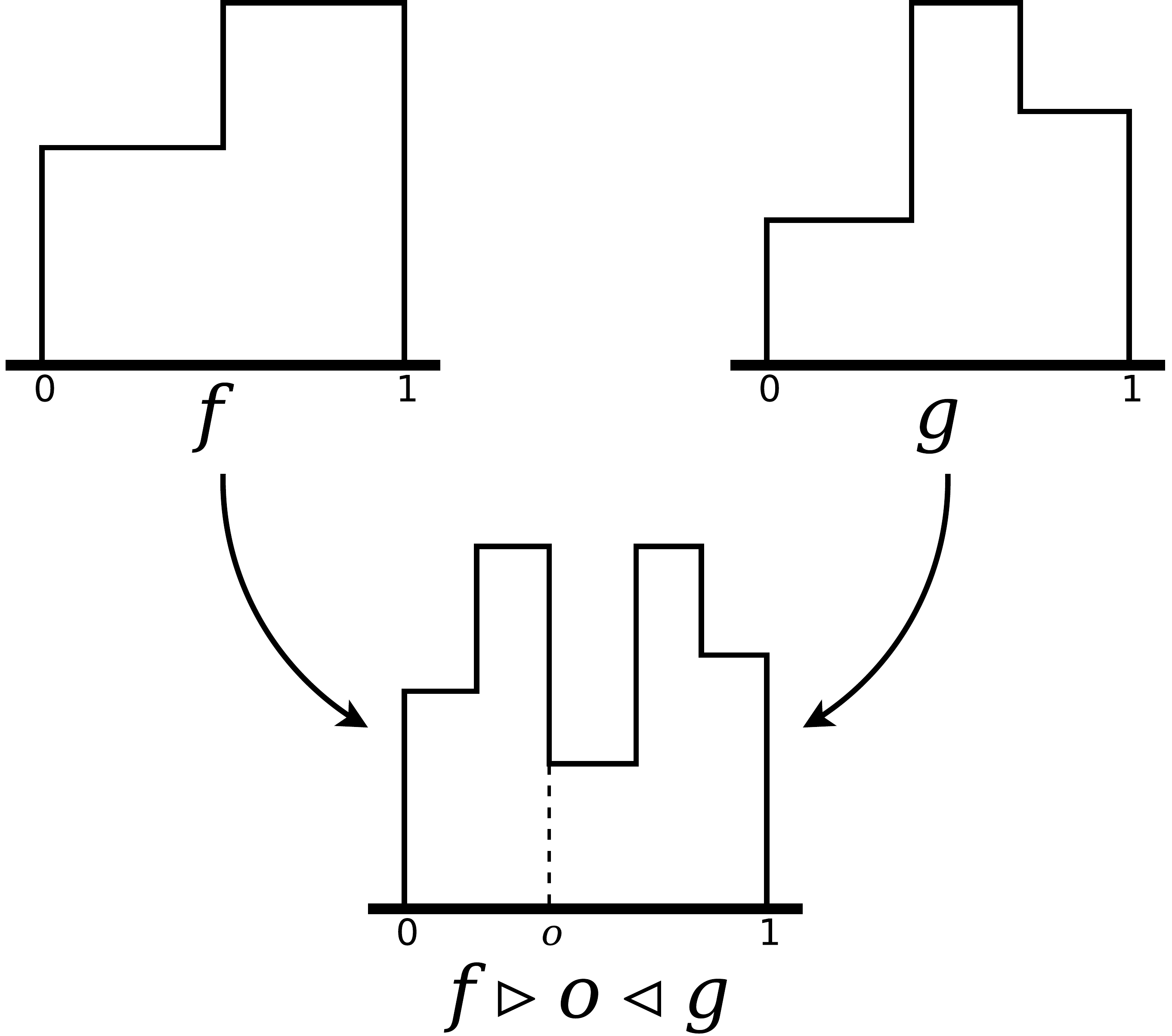}}\end{center}
  \caption{Given two step functions $f$ and $g$, the step function
  $\glue{f}{o}{g}$ is $f$ squeezed into $[0, o]$ and $g$ squeezed into $[o,1]$.}
\end{figure}

Even though we call step functions ``functions'', they are not really
functions, and we never formally interpret them as functions. They are a
formal structure that takes the place of step functions from classical
mathematics. It does not matter that our informal interpretation of
$\glue{f}{o}{g}$ is not well defined at $o$, because the step functions are
intended for integration, not for evaluation at a point.

One can see that this inductive type is a binary tree whose nodes hold data of
type {\ou}, and whose leaves have type $X$. We work with an equivalence
relation on this binary tree structure that identifies different ways of
constructing the same step function. Informally, this is the equivalence
relation induced by our interpretation; the formal equivalence relation is
defined in Section \ref{ss:SFasAF}.

We define two sorts of inverses to \tmtextsf{glue} which we call
left-split and
right-split. Given $f : \mathfrak{S}X$ and $a : \ou$ we define left-split
(written as $\splitl{f}{a} : \mathfrak{S}X$) and right-split
(written as $\splitr{f}{a} : \mathfrak{S}X$) as follows:

\begin{definition}
  \begin{eqnarray*}
    \text{$\splitl{\pure{x}}{a}$} & \defEq & \pure{x}\\
    \splitl{( \glue{f_l}{o}{f_r})}{a} & \defEq & \left\{\begin{array}{ll}
      \splitl{f_l}{\frac{a}{o}}  & \text{(if $a < o$)}\\
      f_l & \text{(if $a = o$)}\\
      \glue{f_l}{\frac{o}{a}}{\splitl{(f_r}{\frac{a - o}{1 - o})}} & \text{(if
      $a > o$)}
    \end{array}\right.\\
    \splitr{\pure{x}}{a} & \defEq & \pure{x}\\
    \splitr{( \glue{f_l}{o}{f_r})}{a} & \defEq & \left\{\begin{array}{ll}
      \glue{( \splitr{f_l}{\frac{a}{o}})}{\frac{o - a}{1 - a}}{f_r}  &
      \text{(if $a < o$)}\\
      f_r & \text{(if $a = o$)}\\
      \splitr{f_r}{\frac{a - o}{1 - o}} & \text{(if $a > o$)} \text{.}
    \end{array}\right.
  \end{eqnarray*}
\end{definition}

Informally, the left split ({\splitl{f}{a}}) takes the portion of $f$ on the
interval $[0,a]$ and scales it up to the full interval $[0,1]$. The right
split ({\splitr{f}{a}}) does the same thing for the portion of $f$ on the
interval $[a,1]$. We have that
\[\glue{\splitl{(f}{a)}}{a}{\splitr{f)}{(a}\asymp f}\]
holds, which means that gluing back the left and right pieces of a
step function split at $a$ returns an equivalent function back. However, this
process does not generally return an identical representation. The formal
definition of the equivalence relation is defined later in Section
\ref{ss:SFasAF}.

The inductive type for step functions has an associated catamorphism which we
call $\tmop{\mathsf{fold}}$.

\begin{definition}
  \begin{eqnarray*}
    \tmop{\mathsf{fold}} & : & (X \Rightarrow Y) \Rightarrow ( \ou \Rightarrow
    Y \Rightarrow Y \Rightarrow Y) \Rightarrow \SF X \Rightarrow Y\\
    \tmop{\mathsf{fold}} \varphi \psi \hat{x} & \assign & \varphi x\\
    \tmop{\mathsf{fold}} \varphi \psi ( \glue{f}{o}{g}) & \assign & \psi o (
    \tmop{\mathsf{fold}} \varphi \psi f) ( \tmop{\mathsf{fold}} \varphi \psi
    g) \text{.}
  \end{eqnarray*}
\end{definition}

This $\tmop{\mathsf{fold}}$ operation is used in many places. For instance, it
is used to define two metrics on step functions (Section~\ref{metric-step}) or
to check whether a property holds globally on $[0,1]$
(Section~\ref{ss:SFasAF}). Not every $\tmop{\mathsf{fold}}$ respects the
equivalence relation on step functions, so we need to prove that each
$\tmop{\mathsf{fold}}$ instance we use respects the equivalence relation.

\subsection{Step functions form a monad}\label{StepF-mon}The step function
type constructor $\SF$ forms a monad similar to the reader monad $\lambda X.
X^{[0,1]}$~{\cite{Wadler92b}}. The $\tmop{\mathsf{return}}$ of {\SF} is the
constant function, $\tmop{\mathsf{map}}$ is defined in the obvious way using
$\tmop{\mathsf{fold}}$, and the $\tmop{\mathsf{join}}$ from $\SF ( \SF X)$ to $\SF X$ is the formal
variant of the $\tmop{\mathsf{join}}$ function from the reader monad,
$\tmop{\mathsf{join}} f z \defEq f z z$, which considers a step function of
step functions as a step function of two inputs and returns the step function
of its diagonal:

\begin{definition}
  \begin{eqnarray*}
    \tmop{\mathsf{join}}  \pure{f} & \defEq & f\\
    \tmop{\mathsf{join}}  \left( \glue{f}{o}{g} \right) & \defEq &
    \glue{\splitl{\tmop{\mathsf{join}} ( \tmop{\mathsf{map}} (\lambda x.x}{o)
    f)}}{o}{\tmop{\mathsf{join}} ( \tmop{\mathsf{map}} (\lambda x.
    \splitr{x}{o}) g)} \text{.}
  \end{eqnarray*}
\end{definition}

Rather than use these monadic functions, we use the applicative functor
interface to this monad.

\subsection{Applicative functors}\label{ss:applicative}Let $\mathfrak{M}$ be a
strong monad. To lift a function $f : X \Rightarrow Y$ to a function
$\mathfrak{M}X \Rightarrow \mathfrak{M}Y$, we use $\tmop{\mathsf{map}} : (X
\Rightarrow Y) \Rightarrow \mathfrak{M}X \Rightarrow \mathfrak{M}Y$. Lifting a
function with two curried arguments is possible using a similar function
$\tmop{\mathsf{map2}}$. However, to avoid having to write a function
$\tmop{\mathsf{map}} \mathit{n}$ for each natural number $n$, one can use the
theory of applicative functors. An consists of a type constructor
$\mathfrak{T}$ and two functions:
\begin{eqnarray*}
  \tmop{\mathsf{pure}} & : & X \Rightarrow \mathfrak{T}X\\
  \tmop{\mathsf{ap}} & : & \mathfrak{T}(X \Rightarrow Y) \Rightarrow
  \mathfrak{T}X \Rightarrow \mathfrak{T}Y
\end{eqnarray*}
The function $\tmop{\mathsf{pure}}$ lifts any value inside the functor. The
$\tmop{\mathsf{ap}}$ function applies a function inside the functor to a value
inside the functor to produce a value inside the functor. We
denote $( \tmop{\mathsf{pure}} x)$ by $\pure{x}$, as was done for monads, and
we denote $( \tmop{\mathsf{ap}} f x)$ by $f \app x$. An applicative functor
must satisfy the following laws~{\cite{mcbride:2008}}:

\begin{center}
  \begin{tabular}{ll}
    $\pure{\id}  \app v \asymp v$ & Identity\\
    $\pure{\comp}  \app u \app v \app w \asymp u \app (v \app w)$ &
    Composition\\
    $\pure{f}  \app  \pure{x}  \asymp  \pure{f x}$ & Homomorphism\\
    $u \app  \pure{y}  \asymp  \pure{\tmop{ev}_y}  \app u$ & Interchange
  \end{tabular}
\end{center}

Where {\comp} and {\id} are the composition and identity combinators
respectively (see Section \ref{ss:Combinator}) and $\tmop{ev}_y \assign
\lambda f. f y$ is the function which evaluates at $y$.

Every strong monad induces the canonical applicative
functor~{\cite{mcbride:2008}} where
\begin{eqnarray*}
  \tmop{\mathsf{pure}} & \assign & \tmop{\mathsf{return}}\\
  f \app x & \assign & \tmop{\mathsf{bind}} (\lambda g. \tmop{\mathsf{map}} g
  x) f.
\end{eqnarray*}
As the name suggests, every applicative functor can be seen as a functor.
Given an applicative functor $\mathfrak{T}$, we define
$\tmop{\mathsf{map}} : (X \Rightarrow Y) \Rightarrow \mathfrak{T}X \Rightarrow \mathfrak{T}Y$ as
\[ \tmop{\mathsf{map}} f x \defEq \hat{f}  \app x. \]
When $\mathfrak{T}$ is generated from a monad, this definition of
\tmtextsf{map} is equivalent to the definition of \tmtextsf{map} associated
with the monad.

\subsection{The step function applicative functor}\label{ss:SFasAF}The
$\tmop{\mathsf{ap}}$ function for
step functions $\mathfrak{S}$ applies a step function of functions to a step
function of argument pointwise. It is formally defined as follows:
\begin{definition}
  \begin{eqnarray*}
    \pure{f} \app \pure{x} & \defEq & \pure{f (x)}\\
    \pure{f} \app ( \glue{x_l}{o}{x_r}) & \defEq & \glue{( \pure{f} \app
    x_l)}{o}{( \pure{f} \app x_r})\\
    \glue{(f_l}{o}{f_r)} \app x & \defEq & \glue{(f_l \app (
    \splitl{x}{o}))}{o}{(f_r \app ( \splitr{x}{o}))} \text{.}
  \end{eqnarray*}
\end{definition}

For step functions $\SF$, we denote $( \tmop{\mathsf{map}} f
x)$ by $f \Map x$. This notation is meant to suggest the similarity with the
composition operation, which is the definition of $\tmop{\mathsf{map}}$ for
the reader monad $\lambda X. X^{[0,1]}$.

\begin{definition}
  The binary version of $\tmop{\mathsf{map}}$ is defined in terms of
$\tmop{\mathsf{map}}$
  and $\tmop{\mathsf{ap}}$.
  \[ \tmop{\mathsf{map} 2} f a b \assign f \Map a \app b. \]
\end{definition}

Higher arity maps can be defined in a similar way; however, we found it more
natural to simply use $\tmop{\mathsf{map}}$ and $\tmop{\mathsf{ap}}$
everywhere.

We will often use $\tmop{\mathsf{map} 2}$ to lift infix operations. Because of
this, we give it a special notation.

\begin{definition}
  If $\circledast$ is some infix operator such that $\lambda x
  y. x \circledast y : X \Rightarrow Y \Rightarrow Z$, then we define
  \[ f \maptwo{\circledast} g \assign (\lambda x y.x \circledast y) \Map f
     \app g, \]
  where $f : \SF X$, $g : \SF Y$, and $f \maptwo{\circledast} g : \SF Z$.
\end{definition}

For example, if $f, g : \SF \mathbbm{Q}$ are rational step functions, then $f
\maptwo{\op{-}} g$ is the pointwise difference between $f$ and $g$ as a
rational step function.

We can lift relations to step functions as well. A relation is simply a
function to $\Prop$, the type of propositions. Thus a binary relation
$\propto$ has a type $\lambda x y. x \propto y : X \Rightarrow {\nobreak} Y
\Rightarrow {\nobreak} \Prop$. If we use $\tmop{\mathsf{map} 2}$, we end up
with an function $\lambda f g. f \maptwo{\op{\propto}} g : \SF X \Rightarrow
\SF Y \Rightarrow \SF \Prop$. The result is not a proposition, but rather a
step function of propositions. Classically, this corresponds to a step
function of Booleans. In other words, $\SF \Prop$ represents a type of step
characteristic functions on $[0,1]$.

Each way of turning a characteristic function into a proposition determines a
different kind of predicate lifting~{\cite{Schroder:2005}}. For our purposes,
we are interested in the one that asks the characteristic function to hold
everywhere. The function $\tmop{\mathsf{fold}}_{\Prop} : \SF
\Prop \Rightarrow \Prop$ does this by folding conjunction over a step
function.

\begin{definition}
  $\tmop{\mathsf{fold}}_{\Prop} \assign \tmop{\mathsf{fold}} ( \id, \lambda o
  p q. p \wedge q)$.
\end{definition}

When this function is composed with \tmtextsf{map2}, the result lifts a
relation to a relation on step functions.

\begin{definition}
  $f \foldmaptwo{\op{\propto}} g \assign
  \tmop{\mathsf{fold}}_{\Prop} (f \maptwo{\op{\propto}} g)$.
\end{definition}

For example, we define equivalence on step functions by lifting the
equivalence relation on $X$.

\begin{definition}
  $f \asymp_{\mathfrak{S}X} g \assign f \foldmaptwo{\op{\asymp_X}} g$.
\end{definition}

Two step functions are equivalent if they are pointwise equivalent everywhere.
Similarly, we define a partial order on step functions by lifting the
inequality relation on $\mathbbm{Q}$.

\begin{definition}
  $f \leq_{\SF \mathbbm{Q}} g \assign f\{\leq_{\mathbbm{Q}}\}g$.
\end{definition}

A step function $f$ is less than a step function $g$ if $f$ is pointwise less
than $g$ everywhere.

\subsection{Two metrics for step functions}\label{metric-step}The step
functions over the rational numbers, $\SF \mathbbm{Q}$, form a metric space in
two ways, with the $L^{\infty}$ metric
and the $L^1$ metric. We first define the two
norms on the step functions.

\begin{definition}
  \begin{eqnarray*}
    \|f\|_{\infty} & \assign & \tmop{\mathsf{fold}}_{\sup} (
    \tmop{\mathsf{abs}}  \Map f)\\
    \|f\|_1 & \assign & \tmop{\mathsf{fold}}_{\tmop{affine}} (
    \tmop{\mathsf{abs}}  \Map f)
  \end{eqnarray*}
  where
  \begin{eqnarray*}
    \tmop{\mathsf{fold}}_{\sup} & \assign & \tmop{\mathsf{fold}}  \id (\lambda
    o x y. \mathsf{\max} x y)\\
    \tmop{\mathsf{fold}}_{\tmop{affine}} & \assign & \tmop{\mathsf{fold}}  \id
    (\lambda o x y. ox + (1 - o) y)
  \end{eqnarray*}
  and $\tmop{\mathsf{abs}} : \mathbbm{Q} \Rightarrow
  \mathbbm{Q}$ is the absolute value function on $\mathbbm{Q}$.
\end{definition}

The function $\tmop{\mathsf{fold}}_{\sup} : \SF \mathbbm{Q} \Rightarrow
\mathbbm{Q}$ returns the supremum of the step function, while the function
$\tmop{\mathsf{fold}}_{\tmop{affine}} : \SF \mathbbm{Q} \Rightarrow
\mathbbm{Q}$ returns the integral of a step function.

Next, the metric distance between two step functions is defined.

\begin{definition}
  \begin{eqnarray*}
    d^{\infty} f g & \assign & \|f \maptwo{\op{-}} g\|_{\infty}\\
    d^1 f g & \assign & \|f \maptwo{\op{-}} g\|_1 .
  \end{eqnarray*}
\end{definition}

Finally, the distance relations are defined in terms of the distance
functions.

\begin{definition}
  \begin{eqnarray*}
    \ball{\SF^{\infty} \mathbbm{Q}}{\varepsilon}{f}{g} & \assign & d^{\infty}
    f g \leq \varepsilon\\
    \ball{\SF^1 \mathbbm{Q}}{\varepsilon}{f}{g} & \assign & d^1 f g \leq
    \varepsilon .
  \end{eqnarray*}
\end{definition}

When we need to be clear which metric space
is being used, we will use the notation $\SF^{\infty} \mathbbm{Q}$ or $\SF^1
\mathbbm{Q}$.

The two fold functions defined in this section are uniformly continuous for
their respective metrics.
\begin{eqnarray*}
  \tmop{\mathsf{fold}}_{\sup} & : & \SF^{\infty} \mathbbm{Q} \arr
  \mathbbm{Q}\\
  \tmop{\mathsf{fold}}_{\tmop{affine}} & : & \SF^1 \mathbbm{Q} \arr
  \mathbbm{Q}
\end{eqnarray*}
The identity function is {\uc} in one direction, $\iota :
\SF^{\infty} \mathbbm{Q} \arr \SF^1 \mathbbm{Q}$; however, the other direction
is not {\uc}.

The metrics $\mathfrak{S}^{\infty} X$ and $\SF^1 X$ can be defined for any
metric space $X$:
\begin{eqnarray}
  \ball{\SF^{\infty} X}{\varepsilon}{f}{g} & \assign &
  \tmop{\mathsf{fold}}_{\Prop} ( \ballsym^X_{\varepsilon} \Map f \app g)
  \nonumber\\
  \ball{\SF^1 X}{\varepsilon}{f}{g} & \assign & \exists h : \SF
  \mathbbm{Q}^{\upl} . \tmop{\mathsf{fold}}_{\Prop} ( \ballsym^X \Map h \app f
  \app g) \wedge \|h\|_1 \leqslant \varepsilon \nonumber
\end{eqnarray}
We have implemented the generic $\mathfrak{S}^{\infty} X$ metric in our
formalization. However, for the $L^1$ space, we have only implemented the
specific $\SF^1 \mathbbm{Q}$ metric.

\subsection{Integrable functions and bounded functions}\label{ss:IFBF}The
bounded functions and the integrable functions are defined as the completion of the step functions
under the $L^{\infty}$ and the $L^1$ metrics respectively.

\begin{definition}
  \begin{eqnarray*}
    \BF & \assign & \C \circ \SF^{\infty}\\
    \IF & \assign & \C \circ \SF^1 .
  \end{eqnarray*}
\end{definition}

In Section~\ref{ss:InductiveDef}, we informally interpreted elements of $\SF
X$ as (partially defined) functions on $[0, 1]$. Similarly, we can informally
interpret each bounded function as a (partially defined) function. Consider $f
: \BF \mathbbm{Q}$. Define $g_n \assign f \left( \frac{1}{n} \right)$. Then
$\underset{n \rightarrow \infty}{\lim} ~ g_n (x)$ exists for all points $x$ in
$[0,1]$ except perhaps for the (rational) splitting points of the step
functions $g_n$. At the points where this limit is defined, it is
(classically) continuous.

To every Riemann integrable function on $[0,1]$ we can associate
an element in $\IF \mathbbm{Q}$. Moreover, functions $f$ and $g$ such that
$\int |f - g| \asymp 0$ will be assigned to equivalent elements in $\IF
\mathbbm{Q}$. This definition can be extended to every {\tmem{generalized}}
Riemann integrable function, where a function $h$ is generalized Riemann
integrable if $h_n \defEq \mathsf{\max} \left( \mathsf{\min} h \hat{n} \right)
\left( \widehat{\um n} \right)$ is integrable for each $n$ and the limit
of$\int h_n$ converges (even though $h_n$ may not converge pointwise
everywhere). Conversely, we can informally interpret every element $f$ of $\IF
\mathbbm{Q}$ as a generalized Riemann integrable function. Define $g_n$ as the
sequence
\[ g_n \defEq f \left( \frac{1}{2^{2 n + 1}} \right) \text{.} \]
By the fundamental lemma of integration~{\cite{lang:1993}}, $g_n$ converges
pointwise almost everywhere. Let $g$ be this pointwise limit. Then $g$ is a
generalized Riemann integrable function associated with $f$.

The bounded functions have a supremum operation,
$\tmop{\mathsf{sup}}: \BF \mathbbm{Q} \arr \R$ and, similarly, the integrable
functions have an integration operation, $\int: \IF \mathbbm{Q} \rightarrow
\R$ which are defined by lifting the two folds from the previous section.

\begin{definition}\label{IntegralQ}
  \begin{eqnarray*}
    \tmop{\mathsf{sup}} f &  \defEq & \fastMap_{\C}  \tmop{\mathsf{fold}}_{\tmop{\mathsf{sup}}}
f\\
    \int f & \defEq & \fastMap_{\C}  \tmop{\mathsf{fold}}_{\tmop{affine}} f
  \end{eqnarray*}
\end{definition}

There is an injection from the bounded functions into the integrable functions
defined by lifting the injection on step functions: $\lift{\iota} : \BF
\mathbbm{Q} \rightarrow \IF \mathbbm{Q}$. However, there is no injection from
integrable functions to bounded functions. Thus bounded functions can be
integrated, but integrable functions may not have a supremum.

\subsection{Riemann integral}\label{ss:riemann}The process for integrating a
function is as follows. Given a function $f$, one needs to find an equivalent
representation of $f$ as an integrable function and then this integrable
function can be integrated. We will consider how to integrate {\uc} functions
on $[0,1]$, which is a useful class of functions to integrate.

We convert a {\uc} function to an integrable function by a two step process.
First, we will convert it to a bounded function, and then the bounded function
can be converted to an integrable function using the injection defined in the
previous section.

To produce a bounded function, one needs to create a step function that
approximates $f$ within $\varepsilon$ for any value $\varepsilon :
\mathbbm{Q}^+$. The usual way of doing this is to create a step function where
each step has width no more than $2 \left( \mu_f \varepsilon \right)$. The
value at each step is taken by sampling the function at the center of the
step.
\begin{figure}[h]\label{fig:Riemann}
  \hspace*{-2.5cm}\includegraphics[height=4.3cm]{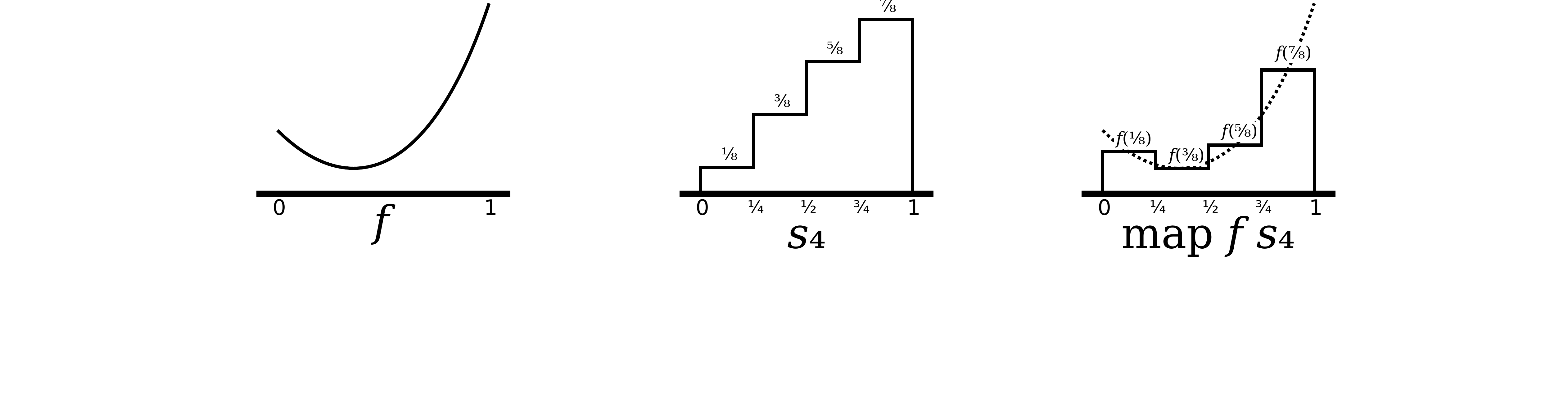}\vspace*{-17mm}
  \caption{Given a uniformly continuous function $f$ and a step function $s_4$
  that approximates the identity function, the step function $(
  \tmop{\mathsf{map}} f s_4)$ (or $f \Map s_4$) approximates $f$ in the
  familiar Riemann way.}
\end{figure}

When developing the above, it became clear that one can
achieve the desired result by creating a step function whose values are the
sample inputs, and then mapping $f$ over these ``sampling step-functions''
(see Figure~\ref{fig:Riemann}). In fact, the limit of these
``sampling step-functions'' is simply the identity function on $[0,1]$
represented as a bounded function, $\idzeroone: \BF \mathbbm{Q}$ (see
Section~\ref{ss:id01}). Given any {\uc} function $f: \mathbbm{Q} \arr
\mathbbm{Q}$, we can prove that $\tmop{\mathsf{map}}_{\SF^{\infty}} f :
\SF^{\infty} \mathbbm{Q} \arr \SF^{\infty} \mathbbm{Q}$ is {\uc}. We can then
lift again to operate on bounded functions, $\fastMap_{\C} \left(
\tmop{\mathsf{map}}_{\SF^{\infty}} f \right) : \BF \mathbbm{Q} \arr \BF
\mathbbm{Q}$. Applying this to {\idzeroone} yields $f$ restricted to $[0,1]$
as a bounded function, which can then be converted to an integrable function
and integrated.

\begin{definition}
  $\int_{[0,1]} f \assign \int \left( \fastMap_{\C} \iota
  \left( \fastMap_{\C} \left( \tmop{\mathsf{map}}_{\SF^{\infty}} f \right)
  \idzeroone \right) \right)$.
\end{definition}

With a small modification, this process will also work for $f : \mathbbm{Q}
\arr \mathbbm{R}$. In this case $\tmop{\mathsf{map}} f$ has type $\SF
\mathbbm{Q} \Rightarrow \SF \mathbbm{R}$, Fortunately, there is an injection
$\tmop{\mathsf{dist}} : \SF \mathbbm{R} \Rightarrow \BF \mathbbm{Q}$, that
interprets a step function of real values as a bounded function (see
Definition~\ref{StepDist}). We can prove that the composition
$\tmop{\mathsf{dist}} \circ ( \tmop{\mathsf{map}}_{\SF} f) : \SF^{\infty}
\mathbbm{Q} \arr \BF \mathbbm{Q}$ is {\uc}. Then, proceeding in a similar
fashion, this can be lifted with {\fastBind} and applied to {\idzeroone} to
yield $f$ restricted to $[0,1]$ as a bounded function, which can then be
integrated.

\begin{definition}
  $\int_{[0,1]} f \assign \int \left(
  \tmop{\mathsf{map}}_{\C} \iota \left( \fastBind_{\C}  \left(
  \tmop{\mathsf{dist}} \circ ( \tmop{\mathsf{map}}_{\SF} f) \right)
  \idzeroone \right) \right)$.
\end{definition}

An arbitrary uniformly continuous function $f:\mathbbm{R} \arr \mathbbm{R}$
can be integrated on $[0,1]$ by integrating $f \circ
\tmop{\mathsf{return}}_{\C} : \mathbbm{Q} \arr \mathbbm{R}$ because the
Riemann integral only depends on the value of functions at rational points.

\subsection{Stieltjes integral}\label{ss:stieltjes}Given the previous
presentation, any bounded function could be used in place of {\idzeroone}. A
natural question arises: what happens when {\idzeroone} is replaced by another
bounded function, $g : \BF \mathbbm{Q}$? An analysis shows that the result is
the Stieltjes integral with respect to $g^{\um 1}$, when $g$ is non-decreasing.

\begin{definition}
  $\int f \mathd g^{\um 1} \assign \int \left(
  \tmop{\mathsf{map}}_{\C} \iota \left( \fastBind_{\C}  \left(
  \tmop{\mathsf{dist}} \circ \left( \tmop{\mathsf{map}}_{\SF} f \right)
  \right) g \right) \right)$.
\end{definition}

We never intended to develop the Stieltjes integral; however, it practically
falls out of our work for free. This is not quite as general as the Stieltjes
integral for three reasons. Because $g$ is defined on $[0,1]$, this means that
$g^{\um 1}$'s range must go from 0 to 1. Essentially, $g^{\um 1}$ must be a
cumulative distribution function and, hence, $g$ is a quantile
function. Secondly, because $g$ is a bounded function, $g^{\um 1}$ must have
compact support (meaning $g^{\um 1}$ must be 0 to the left of its support and
1 to the right of its support). Thirdly, our bounded functions can only have
discontinuities at rational points.

We have tried to allow $g$ to be an arbitrary integrable function (this would
remove some of the previous restrictions); however, we have been unable to
constructively show that $\tmop{\mathsf{dist}} \circ (
\tmop{\mathsf{map}}_{\SF} f) : \SF^1 \mathbbm{Q} \Rightarrow \IF \mathbbm{Q}$
is uniformly continuous when $f$ is. We have generated counterexamples where
$f$ is uniformly continuous with modulus $\mu$ and $\tmop{\mathsf{dist}} \circ
( \tmop{\mathsf{map}}_{\SF} f)$ is {\tmem{not}} uniformly continuous with
modulus $\mu$; however, for our particular counterexamples,
$\tmop{\mathsf{dist}} \circ ( \tmop{\mathsf{map}}_{\SF} f)$ is still uniformly
continuous with a different modulus.

Still, our integral should allow one to integrate with respect to some
interesting distributions such as the Dirac distribution
and the Cantor distribution.

\subsection{Distributing monads}\label{ss:dist-monad}The function
$\tmop{\mathsf{dist}} : \SF \mathbbm{R} \Rightarrow \BF \mathbbm{Q}$ combines two monads on metric
spaces, {\C} and
$\SF$. The function $\tmop{\mathsf{dist}}$ has type $\SF ( \C \mathbbm{Q})
\Rightarrow \C ( \SF \mathbbm{Q})$. In general, the composition of two monads
$\mathfrak{M} \circ \mathfrak{N}$ forms a monad when there is a distribution
function $\tmop{\mathsf{dist}} : \mathfrak{N}(\mathfrak{M}X) \rightarrow
\mathfrak{M}(\mathfrak{N}X)$ satisfying certain laws~{\cite{Beck:1969,TTT}}.
Below we state the laws in a more familiar function style~{\cite{Jones:1993}}:
{\footnote{For the {\SF} and {\C} monads, we formally checked all of these
rules apart from the last one which was too tedious; however, the correctness
of the integral does not depend on the proofs of these laws.}}
\begin{eqnarray}
  \tmop{\mathsf{dist}} \circ \tmop{\mathsf{map}}_{\mathfrak{N}}  \left(
  \tmop{\mathsf{map}}_{\mathfrak{M}} f \right) & \asymp &
  \tmop{\mathsf{map}}_{\mathfrak{M}} \left( \tmop{\mathsf{map}}_{\mathfrak{N}}
  f \right) \circ \tmop{\mathsf{dist}} \nonumber\\
  \tmop{\mathsf{dist}} \circ \tmop{\mathsf{return}}_{\mathfrak{N}} & \asymp &
  \tmop{\mathsf{map}}_{\mathfrak{M}}  \tmop{\mathsf{return}}_{\mathfrak{N}}
  \nonumber\\
  \tmop{\mathsf{dist}} \circ \tmop{\mathsf{map}}_{\mathfrak{N}}
  \tmop{\mathsf{return}}_{\mathfrak{M}} & \asymp &
  \tmop{\mathsf{return}}_{\mathfrak{M}} \nonumber\\
  \tmop{\mathsf{prod}} \circ \tmop{\mathsf{map}}_{\mathfrak{N}}
  \tmop{\mathsf{dorp}} & \asymp & \tmop{\mathsf{dorp}} \circ
  \tmop{\mathsf{prod}} \nonumber
\end{eqnarray}
where
\begin{eqnarray}
  \tmop{\mathsf{prod}} & \defEq & \tmop{\mathsf{map}}_{\mathfrak{M}}
  \tmop{\mathsf{join}}_{\mathfrak{N}} \circ \tmop{\mathsf{dist}} \nonumber\\
  \tmop{\mathsf{dorp}} & \defEq & \tmop{\mathsf{join}}_{\mathfrak{M}} \circ
  \tmop{\mathsf{map}}_{\mathfrak{M}}  \tmop{\mathsf{dist}} . \nonumber
\end{eqnarray}
\begin{definition}
  \label{StepDist}In our case, the distribution function is defined as
  \begin{eqnarray*}
    \tmop{\mathsf{dist}} & : & \SF^{\infty} \left( \C X \right)  \arr \C
    \left( \SF^{\infty} X \right)\\
    \tmop{\mathsf{dist}} f & \defEq & \lambda \varepsilon .
    \tmop{\mathsf{map}}_{\SF^{\infty}} \left( \lambda x.x \varepsilon \right)
    f.
  \end{eqnarray*}
\end{definition}

The function $\tmop{\mathsf{dist}}$ maps a step function $f$ with values in
the completion of $X$ to a collection of approximations $f_{\varepsilon} :
\SF^{\infty} X$ to the function $f$ such that for all $\varepsilon$ in
$\mathbbm{Q}^{\upl}$, $|f - f_{\varepsilon} | \leq \varepsilon$ ``pointwise''.

\section{Implementation in Coq}

In this section, we treat aspects related to our implementation in Coq.

\subsection{Formalization in Coq}\label{ss:FIC}Formalizing the previous in Coq
is done in a straightforward manner. We interpret {\Prop} as
\tmtexttt{Prop}, the universe of propositions. Thus, for example, the ball
relation on rational numbers has type \tmtexttt{Qball : Qpos -> Q -> Q ->
Prop}.

The metric space structure is packaged up as a dependent record, a
$\Sigma$-type. This record contains a field for the domain of the metric
space, which is a setoid, a ball relation over that domain with a proof that
the ball relation respects the equivalence relation of the domain. Lastly the
record contains a collection of proofs of the five axioms of a metric space
(see Section~\ref{ss:MS}) which are themselves packed into their own record
type.

\begin{figure}
\begin{verbatim}
Record is_MetricSpace (X:Setoid)(B: Qpos -> relation X):Prop :=
{ msp_refl: forall e, reflexive _ (B e)
; msp_sym: forall e, symmetric _ (B e)
; msp_triangle: forall e1 e2 a b c, B e1 a b -> B e2 b c ->
                B (e1 + e2)%Qpos a c
; msp_closed: forall e a b,(forall d, B(e+d)%Qpos a b)->B e a b
; msp_eq: forall a b, (forall e, B e a b) -> st_eq a b
}.

Record MetricSpace : Type :=
{ msp_is_setoid :> Setoid
; ball : Qpos -> msp_is_setoid -> msp_is_setoid -> Prop
; ball_wd : forall (e1 e2:Qpos), (QposEq e1 e2) ->
            forall x1 x2, (st_eq x1 x2) ->
            forall y1 y2, (st_eq y1 y2) ->
            (ball e1 x1 y1 <-> ball e2 x2 y2)
; msp : is_MetricSpace msp_is_setoid ball
}.
\end{verbatim}
\caption{The formal definition of a metric space as a dependent record.}
\end{figure}

The completion monad is a function from the record type of metric spaces to
the record type of metric spaces. In Section~\ref{ss:completion-monad} the
domain of the completion is given with an existential quantifier. We use
Coq's \tmtexttt{Set} based existential quantifier (essentially a
$\Sigma$-type) to implement this quantifier.

As a rule, we use \tmtexttt{Prop} based objects only for types that would
(extensionally) have at most one value, these are essentially the Harrop
formulas~{\cite{lcf:spi:03}}. Thus negative types such as function
types/implications whose result type is $\bot$ or $\top$ go into the
\tmtexttt{Prop} universe, and all other types are put into the \tmtexttt{Set}
or \tmtexttt{Type} universes. We chose to have the ball relation return
\tmtexttt{Prop} because the closed sets are typically negative predicates.

Step functions are represented by an inductive data type which is effectively
a labeled binary tree. The Coq declaration for this structure is the
following:
\begin{verbatim}
Inductive StepF : Type:=
|constStepF : X -> StepF
|glue : OpenUnit -> StepF -> StepF -> StepF
\end{verbatim}

Eventually we defined the intended equivalence relation on step functions (see
Section~\ref{ss:SFasAF}) as a binary predicate, but first we define the split
(Section~\ref{glue-split}) and basic applicative functor functions. For
example, \tmtexttt{Ap} is defined as:
\begin{verbatim}
Fixpoint Ap (X Y:Type)(f:StepF (X->Y))(a:StepF X):StepF Y :=
match f with
|constStepF f0 => Map f0 a
|glue o f0 f1=>let (l,r):=Split a o in (glue o(Ap f0 l)(Ap f1 r))
end.
\end{verbatim}

We created proofs of the various laws and relationships between our
definitions. This cumulates with an ultimate proof that our definition of
integration coincides with a previous reference implementation from the CoRN
library~{\cite{lcf:03}}:
\begin{figure}
\begin{verbatim}
Lemma Integrate01_correct : forall F (H01:Zero[<=](One:IR))
 (HF:Continuous_I H01 F) (f:Q_as_MetricSpace --> CR),
 (forall (o:Q) H, (0 <= o <= 1)->
 (f o == IRasCR (F (inj_Q IR o) H)))%CR ->
 (IRasCR (integral Zero One H01 F HF)==Integrate01 f)%CR.
\end{verbatim}\caption{The theorem stating that our definition of integral is correct.}
\end{figure}

Loosely speaking this says ``for any function \tmtexttt{F} over CoRN's real
number which is continuous on $[0,1]$ and for any function \tmtexttt{f} from
the rationals to our real numbers that agrees with \tmtexttt{F} for rational
inputs between 0 and 1, then CoRN's integral of \tmtexttt{F} over $[0,1]$ is
equivalent to our integral of \tmtexttt{f}. The proof of this lemma is 300
lines long and mostly consists of translating facts about the fast
implementation of the reals to the C-CoRN library and vice versa. The actual
proof is quite general because it only uses certain general properties of the
integral, such as linearity and monotonicity.

As a by-product of our development, we can also compute the supremum of any
uniformly continuous function on $[0,1]$.

This has been a small glimpse into our Coq development. For full details
there is no better source than the source; see $\langle$\tmtexttt{http://c-corn.cs.ru.nl}$\rangle$.

\subsection{Glue and split}\label{glue-split}As discussed in
Section~\ref{ss:FIC}, step functions are an inductive structure defined by two
constructors. One constructor creates constant step
functions, and the other constructor, \tmtexttt{glue}, squeezes two step functions
together, joining them together at a given point $o : \ou$. One of the first
operations we defined on step functions (after defining
\tmtexttt{fold}) was \tmtexttt{Split}, which is like the opposite of
\tmtexttt{glue}. Recall from Section~\ref{ss:InductiveDef} that, given a step
function $f$ and a point $a : \ou$, \tmtexttt{Split} splits $f$ into two
pieces at $a$. The functions \tmtexttt{SplitL} and \tmtexttt{SplitR} return the left step function
and
the right step function respectively. Table~\ref{CoqStepFunction} lists the
association between our mathematical notation and the concrete syntax used in
Coq.

\begin{table}[h]
  \begin{center}
  \begin{tabular}{|c|c|}
    \hline
    {\tmstrong{Mathematical Notation}} & {\tmstrong{Coq Syntax}}\\
    \hline
    {\pure{x}} & \tmtexttt{constStepF x}\\
    \hline
    {\glue{f}{o}{g}} & \tmtexttt{glue o f g}\\
    \hline
    {\splitl{f}{a}} & \tmtexttt{SplitL f a}\\
    \hline
    {\splitr{f}{a}} & \tmtexttt{SplitR f a}\\
    \hline
    $( \splitl{f}{a}, \splitr{f}{a})$ & \tmtexttt{Split f a}\\
    \hline
  \end{tabular}\end{center}
  \caption{The concrete syntax used in Coq for our step function notation.\label{CoqStepFunction}}
\end{table}

The key to reasoning about \tmtexttt{Split} was to prove the
\tmtexttt{Split-Split} lemmas:
\begin{eqnarray*}
  ab = c & \Rightarrow & \splitl{\splitl{f}{a}}{b} \asymp \splitl{f}{c}\\
  a + b - ab = c & \Rightarrow & \splitr{\splitr{f}{a}}{b} \asymp
  \splitr{f}{c}\\
  a + b - ab = c \Rightarrow dc = a & \Rightarrow &
  \splitl{\splitr{f)}{(a}}{b} \asymp \splitr{\splitl{(f}{c)}}{d}
\end{eqnarray*}
This collection of lemmas shows how the splits combine and distribute over
each other. With sufficient case analysis, one can prove the above lemmas.
These lemmas, combined with a few other useful lemmas (such as
\tmtexttt{Split-Map} lemmas) provided enough support to prove the laws for
applicative functors without difficulty.

\subsection{Equivalence of step functions}\label{ss:SFEquiv}The work in the
previous section defined an applicative functor of step functions over any
type $X$. From this point on, we will require that $X$ be a setoid (see
Section~\ref{ss:SFasAF}). In order to help facilitate this, in our development
we define new functions, \tmtexttt{constStepF}, \tmtexttt{glue},
\tmtexttt{Split}, etc., that operate on step functions of setoids rather than
step functions of types. These functions are definitionally equal to the
previous functions, but their types now carry the setoid relation from their
argument types to their result types. These new function names shadow the old
function names, and the lemmas about them need to be repeated; however, their
proofs are trivial by using previous proofs.

Perhaps the biggest challenge we encountered in our formalization was to prove
that lifting setoid equivalence to step functions
(Section~\ref{ss:applicative}) is indeed an equivalence relation{\emdash}in
particular showing that it is transitive. We eventually succeeded after
creating some lemmas about the interaction between the equivalence relation
and \tmtexttt{Split}, etc.

\subsection{Common partitions}When reasoning about two (or more) step
functions, it is common to split up one of the step functions so that it
shares the same partition structure as the other step function. This allows
one to do induction over two step functions and have both step functions
decompose the same way. Eventually, we abstracted this pattern of reasoning
into an induction-like principle.
{\scriptsize
\begin{tabular}[t]{llccl}
\tmtexttt{Lemma StepF\_ind2}:\\
$\forall X Y. \forall \Psi : X \Rightarrow Y \Rightarrow \Prop .$ &  &  &  &\\
$(\forall s_0 s_1 t_0 t_1 : \SF X. s_0 \asymp s_1 \Rightarrow t_0 \asymp t_1
\Rightarrow \Psi s_0 t_0 \Rightarrow \Psi s_1 t_1$ &  &  &  & $)
\Rightarrow$\\
$(\forall x : X. \forall y : Y.$ & $\Psi$ & $\pure{x}$ & $\pure{y}$ & $)
\Rightarrow$\\
$(\forall o : \ou . \forall s_l s_r : \SF X. \forall t_l t_r : \SF Y. \Psi
s_l t_l \Rightarrow \Psi s_r t_r \Rightarrow$ & $\Psi$ & $
\glue{s_l}{o}{s_r}$ & $\glue{t_l}{o}{t_r}$ & $)\Rightarrow$\\
$\forall s : \SF X. \forall t : \SF Y.$ & $\Psi$ & $s$ &  $t$ &
\end{tabular}}

This lemma may look complex, but it is as easy to use in Coq as an induction
principle for an inductive family. Normally one would reason about two step
functions by assuming, without loss of generality, that they have a common
partition, then doing induction over that partition. Our lemma above combines
these two steps into one. In one step, one does induction as if the two
functions have a common partition. This lemma was inspired by McBride and
McKinna's work on views in dependent type theory~{\cite{mcbride:2004}}. It
allows one to ``view'' two step functions as having a common partition.

The lemma is used by applying it to a goal of the form \tmtexttt{forall (s t :
StepF X), {\tmem{<expr>}}}, which can be created by generalizing two step
functions. There are only two cases to consider. One case is when \tmtexttt{s}
and \tmtexttt{t} are both constant step functions. The other case is when
\tmtexttt{s} and \tmtexttt{t} are each glued together from two step functions
{\tmem{at the same point}}. There is, however, a side condition to be proved.
One has to show that \tmtexttt{{\tmem{<expr>}}} respects the equivalence
relation on step functions for \tmtexttt{s} and \tmtexttt{t}. Fortunately,
\tmtexttt{{\tmem{<expr>}}} is typically constructed from respectful functions,
and proving this side condition is easy.

For example, we used this lemma in the proof that
$\tmop{\mathsf{fold}}_{\tmop{affine}}$ is additive.

\begin{theorem}
  For all step functions $f, g : \SF \mathbbm{Q}$,
  \[\tmop{\mathsf{fold}}_{\tmop{affine}} f +
  \tmop{\mathsf{fold}}_{\tmop{affine}} g =
  \tmop{\mathsf{fold}}_{\tmop{affine}} (f \maptwo{\op{+}} g)\]
\end{theorem}

\begin{proof}
  The predicate $\lambda f g. \text{$\tmop{\mathsf{fold}}_{\tmop{affine}} f +
  \tmop{\mathsf{fold}}_{\tmop{affine}} g \asymp
  \tmop{\mathsf{fold}}_{\tmop{affine}} (f \maptwo{\op{+}} g)$}$ is a
  respectful predicate because $\tmop{\mathsf{fold}}_{\tmop{affine}}$ and
  addition are respectful functions. Therefore, we can apply
  \tmtexttt{StepF\_ind2}. There are only two cases to consider.
  
  The first case is when $f = \pure{x}$ and $g = \pure{y}$. In this case, the
  problem reduces to $x + y = x + y$ after evaluating
  $\tmop{\mathsf{fold}}_{\tmop{affine}}$ and $\pure{x} \maptwo{\op{+}}
  \pure{y}$.
  
  The second case is when $f = \glue{f_l}{o}{f_r}$ and $g =
  \glue{g_l}{o}{g_r}$. In this case, the problem reduces to
  \begin{eqnarray*}
    & o ( \tmop{\mathsf{fold}}_{\tmop{affine}} f_l +
    \tmop{\mathsf{fold}}_{\tmop{affine}} g_l) + (1 - o) (
    \tmop{\mathsf{fold}}_{\tmop{affine}} f_r +
    \tmop{\mathsf{fold}}_{\tmop{affine}} g_r) & \\
    & = & \\
    & o ( \tmop{\mathsf{fold}}_{\tmop{affine}} (f_l \maptwo{\op{+}} g_l) + (1
    - o) ( \tmop{\mathsf{fold}}_{\tmop{affine}} (f_r \maptwo{\op{+}} g_r)) &
  \end{eqnarray*}
  after evaluation. This then follows from the inductive hypothesis.
\end{proof}

This induction lemma was also very useful for proving the combinator equations
in Section~\ref{ss:Combinator}.

The proof of \tmtexttt{StepF\_ind2} is not very difficult.

\begin{proof}
  Suppose $\Psi$ is a respectful binary predicate on step functions. Suppose
  it also satisfies the two other hypothesis of the lemma. We need to show
  $\forall s t, \Psi s t$. We proceed first by induction on $s$.
  
  Consider the case when $s = \pure{x}$. Now we do induction on $t$. Consider
  the case when $t = \pure{y}$. This is exactly the situation of our first
  hypothesis, so we are done. Consider the case when $t = \glue{t_l}{o}{t_r}$.
  We need to prove $\Psi \pure{x} ( \glue{t_l}{o}{t_r})$ assuming that $\Psi
  \pure{x} t_l$ and $\Psi \pure{x} t_r$ both hold. We know that
  $\glue{\splitl{( \pure{x}}{o)}}{o}{\splitr{\pure{x})}{(o} \asymp \pure{x}}$
  holds, and because $\Psi$ is respectful we can replace $\pure{x}$ using this
  equivalence. Also $\splitl{\pure{x}}{o}$ and $\splitr{\pure{x}}{o}$ both
  reduce to $\pure{x}$ by evaluation. This leaves us with needing to show
  $\Psi ( \glue{\pure{x}}{o}{\pure{x}}) ( \glue{t_l}{o}{t_r})$. This follows
  from our second hypothesis and our two inductive hypotheses.
  
  Now consider the case when $s = \glue{s_l}{o}{s_r}$. We need to prove
  $\forall t. \Psi (s_l \glue{s_l}{o}{s_r}) t$ assuming that $\forall t. \Psi
  s_l t$ and $\forall t. \Psi s_r t$. Again, we know that
  $\glue{\splitl{(t}{o)}}{o}{\splitr{t)}{(o} \asymp t}$ holds, and because
  $\Psi$ is respectful we can replace $t$ using this equivalence. The proof
  proceeds similar to before.
\end{proof}

\subsection{Combinators}\label{ss:Combinator}The
combinators {\comp} and {\id} are preserved by every
applicative functor (see Section~\ref{ss:applicative}). For the applicative
functor $\SF$, all lambda expressions are preserved. To show this, it is
sufficient to show that each of the {\comp}{\flip}{\const}{\diag} combinators
are preserved. These are the combinators defined by:
\begin{itemize}
  \item $\comp f g x \assign f (g x)$ (compose)
  \item $\flip f x y \assign f y x$ (interchange)
  \item $\id x \assign x$ (identity)
  \item $\const x y \assign x$ (discard)
  \item $\diag f x \assign f x x$ (duplicate)
\end{itemize}
The identity combinator is redundant because $\id \asymp \diag \const$, but it
is still useful.

All lambda expressions can be rewritten in a ``point free'' form using these
combinators. Using combinators allows us to reason about the lambda calculus
without worrying about binders, which are notoriously difficult to do by hand.
In fact, it is one of the main issues in the POPLmark
challenge~{\cite{aydemir05mechanized}}.

\begin{theorem}
  The combinators, {\flip}{\const}{\diag}, are preserved by the
  {\SF} monad.
  \begin{eqnarray*}
    \flip  \Map f \app x \app y & \asymp_{\SF X} & f \app y \app x\\
    \const  \Map x \app y & \asymp_{\SF X} & x\\
    \diag  \Map f \app x & \asymp_{\SF X} & f \app x \app x
  \end{eqnarray*}
\end{theorem}

This means that we can lift any function definable with the $\lambda$-calculus
to step functions.

\subsection{Lifting theorems}\label{ss:liftingTheorems}During our development,
we often needed to prove statements like the transitivity of the order
relation on the step functions:
\[ \forall f g h : \SF \Q . f \foldmaptwo{\leq_{\mathbbm{Q}}} g \Rightarrow g
   \foldmaptwo{\leq_{\mathbbm{Q}}} h \Rightarrow f
   \foldmaptwo{\leq_{\mathbbm{Q}}} h \]
We would like to deduce this statement from the transitivity of the
corresponding pointwise relation:
\[ \forall x y z : \mathbbm{Q}. x \leq_{\mathbbm{Q}} y \Rightarrow y
   \leq_{\mathbbm{Q}} z \Rightarrow x \leq_{\mathbbm{Q}} z \]
First, we use a lemma that lifts universal statements about an arbitrary
predicate $R : X \Rightarrow Y \Rightarrow Z \Rightarrow \Prop$ to a universal
statement about step functions:
\[ (\forall x : X. \forall y : Y. \forall z : Z. R x y z) \Rightarrow \forall
   f : \SF X. \forall g : \SF Y. \forall h : \SF Z.
   \tmop{\mathsf{fold}}_{\Prop} (R \Map f \app g \app h) \]
This yields
\[ \forall f g h : \SF \mathbbm{Q}. \tmop{\mathsf{fold}}_{\Prop} ((\lambda x y
   z. x \leq_{\mathbbm{Q}} y \Rightarrow y \leq_{\mathbbm{Q}} z \Rightarrow x
   \leq_{\mathbbm{Q}} z) \Map f \app g \app h) . \]
Next, we would like to ``evaluate'' the lambda expression as ``applied'' to
the step functions $f$, $g$, and $h$. Because $f$, $g$, and $h$ are variables,
we need to symbolically evaluate the expression. We avoid dealing with binders
by converting the lambda expression into the combinator expression
\[ \ess ( \comp \ess ( \comp ( \comp ( \comp \comp ( \op{\Rightarrow}))) (
   \op{\leq_{\mathbbm{Q}}}))) ( \comp ( \flip ( \comp \ess ( \comp ( \comp (
   \op{\Rightarrow})) ( \op{\leq_{\mathbbm{Q}}})))) (
   \op{\leq_{\mathbbm{Q}}})) \Map f \app g \app h \text{,} \]
where $\ess \assign \comp ( \comp ( \comp \diag) \flip) ( \comp \comp)$ and $(
\op{\Rightarrow})$ and $( \op{\leq_{\mathbbm{Q}}})$ are prefix versions of
these infix functions. This substitution is sound because the combinator term
and lambda expression can easily be shown to be extensionally equivalent (by
normalization), and $\tmop{\mathsf{map}}$ and $\tmop{\mathsf{ap}}$ are
well-defined with respect to extensional equality.

We found the required combinator form by using
lambdabot~{\cite{lambdabot}}, a standard tool for
Haskell programmers. It would have been interesting to implement the algorithm
for finding the combinator form of a lambda term in Coq; however, this was not
the aim of our current research.

Now that the lambda term is expressed in combinator form, we can repeatedly
apply the combinator equations from Section~\ref{ss:applicative} and
Section~\ref{ss:Combinator}. These equations are exactly the rules of
``evaluation'' of this expression ``applied'' to step functions. We put these
equations into a database of rewrite rules and used Coq's
\tmtexttt{autorewrite} system as part of a small custom tactic to
automatically reduce this entire expression in one command, yielding
\[ \forall f g h : \SF \mathbbm{Q}. \tmop{\mathsf{fold}}_{\Prop} (f
   \maptwo{\leq_{\mathbbm{Q}}} g \maptwo{\Rightarrow} g
   \maptwo{\leq_{\mathbbm{Q}}} h \maptwo{\Rightarrow} f
   \maptwo{\leq_{\mathbbm{Q}}} h) . \]
Finally, we push the $\tmop{\mathsf{fold}}_{\Prop}$ inside. To do so, we have
proved a lemma which allows us to distribute implication over
$\tmop{\mathsf{fold}}_{\Prop}$:
\[ \forall P Q : \SF ( \Prop) . ( \tmop{\mathsf{fold}}_{\Prop} (P
   \maptwo{\op{\Rightarrow}} Q)) \Rightarrow \tmop{\mathsf{fold}}_{\Prop} P
   \Rightarrow \tmop{\mathsf{fold}}_{\Prop} Q \]
Repeated application of this lemma yields
\[ \forall f g h : \SF \Q . f \foldmaptwo{\leq_{\mathbbm{Q}}} g \Rightarrow g
   \foldmaptwo{\leq_{\mathbbm{Q}}} h \Rightarrow f \foldmaptwo{\leq_Q} h \]
as required.

\subsection{The Identity Bounded Function}\label{ss:id01}

In order to integrate uniformly continuous functions, we compose them with the
identity bounded function to create a bounded function that can be integrated
(see Section~\ref{ss:riemann}). This requires defining the identity bounded
function on $[0, 1]$.

The bounded functions are the completion of step functions under the
$L^{\infty}$ metric. To create a bounded function, we need to generate a step
function within $\varepsilon$ of the identity function for every $\varepsilon
: \mathbbm{Q}^+$. The number of steps used in the approximation will determine
the number of samples of the continuous function $f$ that will be used. For
efficiency, we want the approximation to have the fewest number of steps
possible. Therefore, we defined a function
$\text{\tmtexttt{stepSample}} :
\text{\tmtexttt{positive}} \Rightarrow \SF \mathbbm{Q}$, where
\tmtexttt{positive} is the binary positive natural numbers, such that
$\text{\tmtexttt{stepSample}} n$ produces the best approximation of the
identity function with $n$ steps.

It is unfortunate that the width of each step is computed during integration,
because we know that the result will always be equivalent to $\frac{1}{n}$ for
these particular step functions. Perhaps some other data structure for step
functions could be used that explicitly stores the length of each step.
However, the time spent computing the length of the interval is usually much
smaller that the time it takes to sample the continuous function $f$.

\subsection{Timings}The version of Riemann integration that we implemented
applies to \tmtextit{general} continuous functions and hence has bad
complexity behavior. If we knew more about the function, for instance if it is
differentiable, faster algorithms could be used~{\cite{edalat:1999}}.

\begin{table}[h]
\begin{center}
  \begin{tabular}{|l|l|}
    \hline
    {\tmstrong{Function}} & {\tmstrong{Time}}\\
    \hline
    {\texttt{(answer 3 (Integrate01 Cunit))}} & 0.18s\\
    \hline
    {\texttt{(answer 2 (Integrate01 cos\_uc))}} & 0.52s\\
    \hline
    {\texttt{(answer 3 (Integrate01 cos\_uc))}} & 8.55s\\
    \hline
    {\texttt{(answer 3 (Integrate01 sin\_uc))}} & 7.48s\\
    \hline
  \end{tabular}
\end{center}
  \caption{{\texttt{Time Eval vm\_compute in ...}} carries out the reduction
  using Coq's virtual machine. The expression {\texttt{answer $n$}} asks for an
  answer to within $10^{- n}$. All computations where carried out on an IBM
  Thinkpad X41.}
\end{table}

When extracted to OCaml, the functions run approximately five times faster
when compiled and optimized.

\section{Future and related work}

Many optimizations are possible. Most time is spend on evaluating the function
at many points, as can be seen by comparing the timings for the
$\mathsf{\sin}$ function and the identity function ({\texttt{CUnit}}) which
have the same modulus of continuity and hence the same partition.

Some ways of speeding up the computation of these functions are discussed
in~{\cite{OConnor:real}}. Most notable are:
\begin{itemize}
  \item the use of dyadic rationals;
  
  \item the use of machine integers, (which will enter Coq in the near
  future);
  
  \item the use of forward propagation of errors instead of our a priori
  estimates of convergence~{\cite{BauerKavkler}};
  
  \item the use of parallelism. Our use of maps and folds makes it easy to run
  the algorithm in parallel. In fact, adding parallelism to the extacted
  O'Caml code by hand speeds up the evalutation by a factor three on a four
  processor machine. This only required making a single function,
  {\texttt{DistrComplete}} (a fold), be evaluated in parallel.
  
  We hope that the technology of parallel functional programming will
  included in Coq in the future.
\end{itemize}
Because of the way that we have defined uniform continuity, one modulus of
continuity applies to an entire function. Even for those parts of the domain
where the function changes slowly, we still must approximate the input to the
same precision that is needed for those parts where the function changes
quickly. This reduces performance somewhat for evaluation of these functions
(at the segments where the function changes slowly), but this causes
particularly bad performance for integration.

Because we only have a global modulus of continuity, we must use uniform
partitions when creating an integrable function from a {\uc} function. This
means that the function is sampled just as often where the function changes
slowly as where the function changes quickly. This uniform sampling can be
quite expensive for integration.

There is some potential to increase efficiency by using a ``non-uniform''
definition of uniform continuity. That is to say, using a definition of
uniform continuity that allows different segments of the domain to have local
moduli associated with them. Ulrich Berger uses such a definition of uniform
continuity to define integration~{\cite{berger:2009}}. Simpson also defines an
integration algorithm that uses a local modulus for a function that is
computed directly from the definition of the function~{\cite{simpson:1998}}.
However, implementing his algorithm directly in Coq is not possible because it
relies on bar induction, which is not available in Coq
unless one adds an axiom such as bar induction to it or one treats the real
numbers as a formal space~{\cite{Sambin:1987}}{\cite{Bauer}}.

The constructive real numbers have already been used to provide a
semi-decision procedure for inequalities of real numbers. Not only for the
constructive real numbers, but also for the non-computational real numbers in
the Coq standard library~{\cite{cekp8}}. The same technique can be applied
here.

Previously, the CoRN project~{\cite{corn}} showed that the formalization of
constructive analysis in a type theory is feasible. However, the extraction of
programs from such developments is difficult~{\cite{lcf:spi:03}}. On the
contrary, in the present article we have shown that if one takes an
algorithmic attitude from the start it {\tmem{is}} possible to obtain feasible
programs.

\section{Conclusions}

We have implemented Riemann integration in constructive mathematics based on
type theory. Type checking guarantees that the implementation meets its formal
specification. The use of the completion and the step function monads helped
to structure the program/proof, as did the use of applicative functors.

Building on the previous implementation of the completion of a metric
space~{\cite{OConnor:real}} and the library~{\cite{lcf:04}}, the current
implementation was completed in four man-months. The program/proof consists of
1155 lines of specifications, 3380 lines of proof, and 170,137 total
characters. The size of the gzipped tarball ({\texttt{gzip -9}}) of all the
source files is 37,039 bytes, which is an estimate of the information content.

Together with the work in~{\cite{OConnor:mscs,OConnor:real,OConnor:compact}},
the current project may be seen as the beginning of the realization of
Bishop's program to use constructive mathematics, based on type theory, as a
programming language for exact analysis.

\section{Acknowledgements}

We thank Cezary Kaliszyk for helping us to implement parallelism in OCaml.

\bibliographystyle{alpha}\bibliography{Riemann.bib,thesis.bib}

\newcommand{\etalchar}[1]{$^{#1}$}
\begin{thebibliography}{CFGW04}

\bibitem[ABF{\etalchar{+}}05]{aydemir05mechanized}
B.~Aydemir, A.~Bohannon, M.~Fairbairn, J.~Foster, B.~Pierce, P.~Sewell,
  D.~Vytiniotis, G.~Washburn, S.~Weirich, and S.~Zdancewic.
\newblock Mechanized metatheory for the masses: The {POPLmark} challenge.
\newblock In {\em Proceedings of the Eighteenth International Conference on
  Theorem Proving in Higher Order Logics (TPHOLs 2005)}, 2005.

\bibitem[Bau08]{Bauer}
Andrej Bauer.
\newblock Efficient computation with dedekind reals.
\newblock Extended abstract for CCA2008, 2008.

\bibitem[BC04]{BC04}
Yves Bertot and Pierre Cast\'eran.
\newblock {\em Interactive Theorem Proving and Program Development. Coq'Art:
  The Calculus of Inductive Constructions}.
\newblock Texts in Theoretical Computer Science. Springer Verlag, 2004.

\bibitem[BCP03]{Capretta}
Gilles Barthe, Venanzio Capretta, and Olivier Pons.
\newblock Setoids in type theory.
\newblock {\em J. Funct. Programming}, 13(2):261--293, 2003.
\newblock Special issue on ``Logical frameworks and metalanguages''.

\bibitem[Bec69]{Beck:1969}
J.~Beck.
\newblock Distributive laws.
\newblock In B.~Eckman, editor, {\em Seminar on Triples and Categorical
  Homology Theory}, number~80 in Lecture Notes in Mathematics, pages 119--140.
  Springer, Berlin, 1969.

\bibitem[Ber09]{berger:2009}
Ulrich Berger.
\newblock {From coinductive proofs to exact real arithmetic}.
\newblock In {\em Computer Science Logic}, pages 132--146. Springer, 2009.

\bibitem[Bis67]{Bishop67}
Errett~A. Bishop.
\newblock {\em {Foundations of constructive analysis}}.
\newblock McGraw-Hill Publishing Company, Ltd., 1967.

\bibitem[Bis70]{Bishop:num}
Errett Bishop.
\newblock Mathematics as a numerical language.
\newblock In {\em Intuitionism and Proof Theory (Proceedings of the summer
  Conference at Buffalo, N.Y., 1968)}, pages 53--71. North-Holland, Amsterdam,
  1970.

\bibitem[BJ06]{lambdabot}
Andrew~J. Bromage and Thomas J{\"a}ger.
\newblock Lambdabot.
\newblock \url{http://www.cse.unsw.edu.au/~dons/lambdabot.html}, 2006.

\bibitem[BK09]{BauerKavkler}
Andrej Bauer and Iztok Kavkler.
\newblock A constructive theory of continuous domains suitable for
  implementation.
\newblock {\em Ann. Pure Appl. Logic}, 159(3):251--267, 2009.

\bibitem[BW05]{TTT}
Michael Barr and Charles Wells.
\newblock Toposes, triples and theories.
\newblock {\em Repr. Theory Appl. Categ.}, (12):x+288 pp., 2005.
\newblock Corrected reprint of the 1985 original [MR0771116].

\bibitem[CF03]{lcf:03}
L.~Cruz-Filipe.
\newblock A constructive formalization of the fundamental theorem of calculus.
\newblock In H.~Geuvers and F.~Wiedijk, editors, {\em Types for Proofs and
  Programs}, volume 2646 of {\em LNCS}, pages 108--126. Springer--Verlag, 2003.

\bibitem[CF04]{lcf:04}
L.~Cruz-Filipe.
\newblock {\em Constructive Real Analysis: a Type-Theoretical Formalization and
  Applications}.
\newblock PhD thesis, University of Nijmegen, April 2004.

\bibitem[CFGW04]{corn}
L.~Cruz-Filipe, H.~Geuvers, and F.~Wiedijk.
\newblock C-corn: the constructive coq repository at nijmegen.
\newblock In A.~Asperti, G.~Bancerek, and A.~Trybulec, editors, {\em
  Mathematical Knowledge Management, Third International Conference, MKM 2004},
  volume 3119 of {\em LNCS}, pages 88--103. Springer--Verlag, 2004.

\bibitem[CFS03]{lcf:spi:03}
L.~Cruz-Filipe and B.~Spitters.
\newblock Program extraction from large proof developments.
\newblock In D.~Basin and B.~Wolff, editors, {\em Theorem Proving in Higher
  Order Logics, 16th International Conference, TPHOLs 2003}, volume 2758 of
  {\em LNCS}, pages 205--220. Springer--Verlag, 2003.

\bibitem[CH88]{CoquandHuet}
Thierry Coquand and G{\'e}rard Huet.
\newblock The calculus of constructions.
\newblock {\em Inform. and Comput.}, 76(2-3):95--120, 1988.

\bibitem[Coe04]{Coen:2004}
Claudio~Sacerdoti Coen.
\newblock A semi-reflexive tactic for (sub-)equational reasoning.
\newblock In Jean-Christophe Filli{\^a}tre, Christine Paulin-Mohring, and
  Benjamin Werner, editors, {\em TYPES}, volume 3839 of {\em Lecture Notes in
  Computer Science}, pages 98--114. Springer, 2004.

\bibitem[CP90]{CoquandPaulin}
Thierry Coquand and Christine Paulin.
\newblock Inductively defined types.
\newblock In {\em C{OLOG}-88 ({T}allinn, 1988)}, volume 417 of {\em Lecture
  Notes in Comput. Sci.}, pages 50--66. Springer, Berlin, 1990.

\bibitem[Eda99]{edalat:1999}
Abbas Edalat.
\newblock Numerical integration with exact real arithmetic.
\newblock In {\em Automata, Languages and Programming, 26th International
  Colloquium, ICALP99, Prague, Czech 227 Republic, July 11-15, 1999,
  Proceedings, volume 1644 of Lecture Notes in Computer Science}, pages
  90--104. Springer, 1999.

\bibitem[GL02]{Compiler}
Benjamin Gr{\'e}goire and Xavier Leroy.
\newblock A compiled implementation of strong reduction.
\newblock In {\em ICFP}, pages 235--246, 2002.

\bibitem[GNSW07]{typesreal-article}
Herman Geuvers, Milad Niqui, Bas Spitters, and Freek Wiedijk.
\newblock Constructive analysis, types and exact real numbers (overview
  article).
\newblock {\em Mathematical Structures in Computer Science}, 17(1):3--36, 2007.

\bibitem[Hof97]{Hofmann}
Martin Hofmann.
\newblock {\em Extensional constructs in intensional type theory}.
\newblock CPHC/BCS Distinguished Dissertations. Springer-Verlag London Ltd.,
  London, 1997.

\bibitem[JD93]{Jones:1993}
Mark~P. Jones and Luke Duponcheel.
\newblock Composing monads.
\newblock Technical Report YALEU/DCS/RR-1004, Yale University, 1993.

\bibitem[KO08]{cekp8}
Cezary Kaliszyk and Russell O'Connor.
\newblock Computing with classical real numbers.
\newblock Submitted for publication to the Journal of Automated Reasoning,
  2008.

\bibitem[Lan93]{lang:1993}
Serge Lang.
\newblock {\em Real and Functional Analysis}.
\newblock Springer, 1993.

\bibitem[ML82]{CMCP}
Per Martin-L\"{o}f.
\newblock Constructive mathematics and computer programming.
\newblock In {\em Logic, methodology and philosophy of science, VI (Hannover,
  1979)}, volume 104 of {\em Stud. Logic Found. Math.}, pages 153--175.
  North-Holland, Amsterdam, 1982.

\bibitem[ML98]{ITT}
Per Martin-L\"{o}f.
\newblock An intuitionistic theory of types.
\newblock In {\em Twenty-five years of constructive type theory (Venice,
  1995)}, volume~36 of {\em Oxford Logic Guides}, pages 127--172. Oxford Univ.
  Press, 1998.

\bibitem[MM04]{mcbride:2004}
Conor McBride and James McKinna.
\newblock The view from the left.
\newblock {\em Journal of Functional Programming}, 14(1):69--111, 2004.

\bibitem[Mog89]{moggi:1989}
E.~Moggi.
\newblock Computational lambda-calculus and monads.
\newblock In {\em Proceedings of the Fourth Annual Symposium on Logic in
  computer science}, pages 14--23, Piscataway, NJ, USA, 1989. IEEE Press.

\bibitem[MP08]{mcbride:2008}
Conor McBride and Ross Paterson.
\newblock Applicative programming with effects.
\newblock {\em J. Funct. Program.}, 18(1):1--13, 2008.

\bibitem[NPS90]{NPS}
Bengt Nordstr{\"o}m, Kent Petersson, and Jan~M. Smith.
\newblock {\em Programming in {M}artin-{L}\"of's type theory}, volume~7 of {\em
  International Series of Monographs on Computer Science}.
\newblock The Clarendon Press Oxford University Press, New York, 1990.
\newblock An introduction.

\bibitem[O'C07]{OConnor:mscs}
Russell O'Connor.
\newblock A monadic, functional implementation of real numbers.
\newblock {\em Mathematical Structures in Computer Science}, 17(1):129--159,
  2007.

\bibitem[O'C08a]{OConnor:real}
Russell O'Connor.
\newblock Certified exact transcendental real number computation in {C}oq.
\newblock In Otmane Ait-Mohamed, editor, {\em TPHOLs}, volume 5170 of {\em
  Lecture Notes in Computer Science}, pages 246--261. Springer, 2008.

\bibitem[O'C08b]{OConnor:compact}
Russell O'Connor.
\newblock A computer verified theory of compact sets.
\newblock In Bruno Buchberger, Tetsuo Ida, and Temur Kutsia, editors, {\em SCSS
  2008}, number 08-08 in RISC-Linz Report Series, pages 148--162, Castle of
  Hagenberg, Austria, July 2008. RISC.

\bibitem[Ric08]{Richman:2008}
Fred Richman.
\newblock Real numbers and other completions.
\newblock {\em Math. Log. Q.}, 54(1):98--108, 2008.

\bibitem[Sam87]{Sambin:1987}
Giovanni Sambin.
\newblock Intuitionistic formal spaces - a first communication.
\newblock In D.~Skordev, editor, {\em Mathematical logic and its Applications},
  pages 187--204. Plenum, 1987.

\bibitem[Sch05]{Schroder:2005}
L.~Schr{\"o}der.
\newblock Expressivity of coalgebraic modal logic: The limits and beyond.
\newblock In V.~Sassone, editor, {\em Foundations of Software Science and
  Computational Structures}, number 3441 in Lecture Notes in Mathematics, pages
  440--454. Springer, Berlin, 2005.

\bibitem[Sch08]{Schwichtenberg}
Helmut Schwichtenberg.
\newblock Realizability interpretation of proofs in constructive analysis.
\newblock {\em Theor. Comp. Sys.}, 43(3):583--602, 2008.

\bibitem[Sim98]{simpson:1998}
Alex~K. Simpson.
\newblock Lazy functional algorithms for exact real functionals.
\newblock {\em Lecture Notes in Computer Science}, 1450:456--464, 1998.

\bibitem[SU98]{Sorensen}
M.~S\"orensen and P.~Urzyczyn.
\newblock {\em Lectures on the {C}urry-{H}oward isomorphism}.
\newblock Elsevier, 1998.

\bibitem[Tea08]{Coq}
The Coq~Development Team.
\newblock {\em The {C}oq Proof Assistant Reference Manual}.
\newblock INRIA-Rocquencourt, 2008.

\bibitem[Tho91]{Thompson:1991}
S.~Thompson.
\newblock {\em Type Theory and Functional Programming}.
\newblock Addison Wesley, 1991.

\bibitem[Wad92a]{Wadler92a}
P.~Wadler.
\newblock Monads for functional programming.
\newblock In {\em {Proceedings of the Marktoberdorf Summer School on Program
  Design Calculi}}, August 1992.

\bibitem[Wad92b]{Wadler92b}
Philip Wadler.
\newblock Comprehending monads.
\newblock {\em Mathematical Structures in Computer Science}, 2(4):461--493,
  1992.

\end{thebibliography}

\end{document}